\documentclass[11pt]{article}

\usepackage{amssymb}
\usepackage{amsfonts}
\usepackage{amsmath}
\usepackage{amsthm}
\usepackage{latexsym}
\usepackage[usenames, dvipsnames]{xcolor}
\usepackage{dsfont}
\usepackage{verbatim}
\usepackage{graphicx}
\graphicspath{ {./Simulationen/} }

\allowdisplaybreaks[4]

\numberwithin{equation}{section}

\newtheorem{theorem}{Theorem}[section]
\newtheorem{lemma}[theorem]{Lemma}
\newtheorem{remark}[theorem]{Remark}

\newtheorem{proposition}[theorem]{Proposition}
\newtheorem{definition}[theorem]{Definition}

\newcommand{\cF}{\mathcal{F}}

\newcommand{\EE}{\mathbb{E}}
\newcommand{\var}{\mathbb{V}\mathrm{ar}}
\newcommand{\cov}{\mathbb{C}\mathrm{ov}}
\newcommand{\PP}{\mathbb{P}}

\newcommand{\cB}{\mathcal{B}}

\newcommand{\cP}{\mathcal{P}}
\newcommand{\bone}{\mathbf{1}}

\newcommand{\NN}{\mathbb{N}}

\newcommand{\RR}{\mathbb{R}}
\newcommand{\CC}{\mathbb{C}}
\newcommand{\FF}{\mathbb{F}}
\newcommand{\dd}{\operatorname{d}}
\newcommand{\fatpi}{\boldsymbol{\pi}}

\newcommand{\fatone}{\mathbf{1}}
\newcommand{\diag}{\,\mathbf{diag}}
\newcommand{\fatQ}{\mathbf{Q}}
\newcommand{\frace}{\mathfrak{E}}

\usepackage{marginnote}
\begin{document}
\author{Anita Behme\thanks{Technische Universit\"at Dresden,
		Institut f\"ur Mathematische Stochastik, 01062 Dresden, Germany, e-mail: anita.behme@tu-dresden.de}}
\title{Volatility modeling in a Markovian environment: \\Two Ornstein-Uhlenbeck-related approaches}
\maketitle
\begin{abstract}
 We introduce generalizations of the COGARCH model of Klüppelberg et al. from 2004 and the volatility and price model of Barndorff-Nielsen and Shephard from 2001 to a Markov-switching environment. These generalizations allow for exogeneous jumps of the volatility at times of a regime switch. Both models are studied within the framework of Markov-modulated generalized Ornstein-Uhlenbeck processes which allows to derive conditions for stationarity, formulas for moments, as well as the autocovariance structure of volatility and price process. It turns out that both models inherit various properties of the original models and therefore are able to capture basic stylized facts of financial time-series such as uncorrelated log-returns, correlated squared log-returns and non-existence of higher moments in the COGARCH case.
\end{abstract}




\section{Introduction}
\setcounter{theorem}{0} 

While the famous Black-Scholes model models financial price processes as stochastic exponentials of Brownian motions, nowadays, it is a standard approach in financial modeling to consider price processes that depend on an underlying stochastic volatility process that exhibits jumps.

A prominent continuous-time model of this type is the stochastic volatility model introduced by Barndorff-Nielsen and Shephard \cite{Barn:Shep:2001} in 2001. In this BNS model, the squared volatility process $V$ and the log asset price $G$ are defined to satisfy the equations 
\begin{align}\label{eq-BNS}
\dd V_t &= -\lambda V_t\,\dd t + \dd L_{\lambda t},\\
\dd G_t &=  (\mu+\beta V_t) \dd t + \sqrt{V_t} \,\dd W_t + \rho \,\dd\widetilde L_{\lambda t}, \quad t\geq 0, \nonumber
\end{align}
where $\lambda>0$, $\mu,\beta\in\RR$, $\rho\in\RR$ is a leverage parameter, $L=(L_t)_{t\ge0}$ is a non-decreasing L\'evy process with centered version $\widetilde L_t=L_t-\EE[L_t]$, and $W=(W_t)_{t\ge0}$ is a standard Brownian motion independent of $L$. The volatility process $V$ is assumed to be the stationary solution of \eqref{eq-BNS}. This implies in particular that it is a special case of a stationary generalized Ornstein-Uhlenbeck (GOU) process, i.e. of a stationary solution to an SDE of the form
	\begin{align}\label{MMGOUSDE}
		\dd V_t=V_{t-}\dd U_t+ \dd K_t,\qquad t\geq 0, 
	\end{align}
for a bivariate Lévy process $(U_t,K_t)_{t\geq 0}$ with $U$ having no jumps of size less or equal to $-1$, see e.g. \cite{BLM10}.

Another, mathematically closely related, continuous-time model for a stochastic volatility process with jumps is given by the COGARCH~model. This can be seen as continuous-time version of the celebrated GARCH(1,1) model, where the squared volatility and price time series are supposed to solve
\begin{align} \label{eq-GARCH}
	V_n & = \beta + \lambda G_{n-1}^2 + \delta V_{n-1}, \\
	G_n&= \epsilon_n \sqrt{V_n}, \quad n\in\NN, \nonumber 
\end{align}
where $\beta>0$, $\lambda, \delta\geq 0$, and for some i.i.d. noise $(\epsilon_n)_{n\in\NN}$ with $\EE[\epsilon_n]=0,$ and $\var(\epsilon_n)=1$. By embedding this model in a continuous-time setting and replacing the i.i.d. noise by jumps of a Lévy process, in 2004 Klüppelberg et al. \cite{KLM2004} derived the COGARCH model where the squared volatility process $V$ and the log asset price are given by 
\begin{align} \label{eq-COGARCH}
\dd V_t =&  V_{t-} \Big(\log \delta \dd t + \frac{\lambda}{\delta} \dd[L,L]^d_t\Big) + \beta \dd t,\\
\dd G_t =& \sqrt{V_{t-}} \,\dd L_t, \quad t\geq 0,\nonumber
\end{align}
where $\beta>0, \lambda\geq 0$, $0<\delta<1$, $L$ is an arbitrary  L\'evy process with non-zero jump measure and $[L,L]^d$ is the pure jump part of the quadratic variation process of $L$, cf. \cite[Chapter II.6]{Protter}.  The volatility process $V$ is assumed to be the stationary solution of \eqref{eq-COGARCH}, and therefore, again, it is a special case of a stationary GOU process. \\

Due to the fact that both continuous-time models mentioned above use a special case of a GOU process as volatility model, they share many properties; see \cite{KLM06} for a detailed comparison of the two approaches. In particular, both models have in common that jump sizes in volatility and price exhibit a fixed deterministic relationship; cf.~\cite{JKMc}. This, however, is not very realistic when considered on a large time-scale. Several attempts to overcome this drawback have been made, e.g. by defining multifactor models as superpositions; cf.~\cite{BehmeChong, BNsuper,BN-S:2011}. In this article we choose a different approach and  consider volatility models in a random environment, where the environment is modeled by a continuous-time Markov chain. Dating back to the 80's, cf. \cite{hamilton}, Markov-switching models have already proven to be a reasonable tool in finance and other areas; see e.g. \cite{AsmussenGram, Buffington, Deelstra, Hainaut, Hamiltonbook, JiangPistorius, Momeya}, the review article \cite{AngTimmermann}, and many others. So, already in 2001, in \cite{FrancqRoussignol} a Markov-switching GARCH($p,q$) (MSGARCH($p,q$)) model has been introduced. Here, $p$ is the order of the GARCH terms (denoted as $V_n$) and $q$ denotes the order of the ARCH noise (here $G_n$). In the MSGARCH(1,1) case this generalizes \eqref{eq-GARCH} to 
\begin{align}\label{eq-MSGARCH}
		V_n & = \beta(M_{n-1}) + \lambda(M_{n-1}) G_{n-1}^2 + \delta (M_{n-1}) V_{n-1}, \\
	G_n&= \epsilon_n \sqrt{V_n}, \quad n\in\NN, \nonumber 
\end{align}
where $(M_n)_{n\in\NN_0}$ is a stationary, irreducible, aperiodic Markov chain on a finite set $\{1,\ldots, N\}$, $\beta(k)>0$, $\lambda(k), \delta(k)\geq 0$, for $k=1,\ldots, N$, and  $(\epsilon_n)_{n\in\NN}$ is an i.i.d. noise sequence, independent of $(M_n)_{n\in\NN_0}$, with $\EE[\epsilon_n]=0,$ and $\var(\epsilon_n)=1$, see also \cite[Section 12.2.2]{FZ2} for details. This MSGARCH model has been shown to be efficient in capturing varying volatility states appearing in data by \cite{BenSaida}. Recently, also in continuous-time first attempts to embed the COGARCH in a randomly switching environment have been made e.g. in \cite{LiLiu} and \cite{MbaMwambi}, where the latter article proposes an application of the resulting model for cryptocurrency portfolio selection. The approach used in both sources relies on concatenations of COGARCH processes with the switching mechanism modeled by a continuous-time Markov chain. \\

Still, despite the success of Markov-switching model in finance, no general attempt to study properties of Markov-switching versions of the BNS and the COGARCH model has been made. In this paper we therefore follow the approach of Klüppelberg et al. to introduce a continuous-time model of MSCOGARCH type. As it turns out, the obtained model nicely fits into the framework of recently introduced Markov-switching extensions of GOU processes; cf. \cite{BS_MMGOU, HuangetalMMOU}. This in turn allows for a rather direct extension of various known properties of the COGARCH model to its Markov-switching counterpart. The approach via Markov modulated GOU processes moreover generalizes the previous approaches in \cite{LiLiu, MbaMwambi} and allows for a broader class of processes. In particular we incorporate exogeneous jumps in the volatility at times of regime switches into our model. This allows to comprise the often observed stylized fact, that conditional volatility tends to jump upwards substantially at the onset of a turbulent period (cf. \cite{dueker}), into the model.\\
In the subsequent Section \ref{S-MSBNS} we define a Markov-switching counterpart of the BNS model \eqref{eq-BNS}. Also in this case, the resulting volatility process naturally turns out to be a Markov modulated GOU process, again allowing for a quick derivation of many basic properties and stylized features of volatility and price process such as e.g. uncorrelated log-returns and correlated squared log-returns.\\
The article closes with a short discussion and outlook in Section \ref{S-end}.

\section{Preliminaries}
\setcounter{theorem}{0} 

\subsection{Markov additive processes}

Throughout this article, let $(\Omega, \mathcal{F}, \FF=(\mathcal{F}_t)_{t\geq 0},\PP)$ be a filtered probability space and let $S=\{1,2,\ldots, N\}$ be a finite set. A Markov process $(X,J)=(X_t,J_t)_{t\geq 0}$ on $\RR^d\times S$, $d\geq 1$, is called a \emph{($d$-dimensional) Markov additive process with respect to the filtration $\FF$ ( $\FF$-MAP)}, if for all $s,t\geq 0$ and for all bounded and measurable functions $f:\RR^d\to\RR$, $g:S\to \RR$
\begin{align}\label{MAPdefinition}
	\EE_{J_0}\left[f(X_{s+t}-X_s)g(J_{s+t})|\mathcal{F}_s\right]=\EE_{J_s}\left[f(X_t-X_0)g(J_t)\right].
\end{align}
Hereby, for any $j\in S$ we write $\PP_j(\cdot):=\PP(\cdot|J_0=j)$ and $\EE_j[\cdot]$ for the expectation with respect to $\PP_j$. 
If not stated otherwise we assume $\FF$ to be the smallest filtration that includes the natural filtration induced by $(X,J)$, and satisfies the \emph{usual hypotheses} of right-continuity and completeness, see e.g. \cite{Protter}. In this case we call $\FF$ the \emph{augmented natural filtration} induced by $(X,J)$ and simply call $(X,J)$ a MAP.\\
Intuitively, MAPs can be understood as switching Lévy processes with additional jumps at times of regime switches. More precisely, since $S$ has been chosen to be finite, the defining property \eqref{MAPdefinition} of the MAP $(X,J)$ implies (cf.~\cite[Chapter XI.2]{Asmussen_AppliedProbandQueues}) that there exists a sequence of $N$ independent $\RR^d$-valued Lévy processes $\lbrace X^{(j)},j\in S\rbrace$ such that, whenever $J_t=j$ on some time interval $(t_1,t_2)$, the additive component $(X_t)_{t_1<t<t_2}$ of the MAP $(X,J)$ behaves in law as $X^{(j)}$.\\
Recall at this point that any Lévy process $(X^{(j)}_t)_{t\geq 0}$ can be uniquely determined by its characteristic triplet $(\gamma_{X^{(j)}},\Sigma^2_{X^{(j)}},\nu_{X^{(j)}})$, where $\gamma_{X^{(j)}}$ denotes the \emph{location parameter}, $\Sigma^2_{X^{(j)}}$ the \emph{Gaussian covariance (matrix)}, and $\nu_{X^{(j)}}$ the \emph{Lévy measure} of the process. \\
We denote the jump times of the background driving Markov chain $J$ of the MAP $(X,J)$ by $\{T_n, n\in\NN\}$. Whenever $J$ jumps at a time, say, $T_k$ from state $i$ to state $j$, it induces an additional jump $Z^{ij}_{X,k}$ for $X$, whose distribution $F_{X}^{ij}$ depends only on $(i,j)$ and neither on the jump time nor on the jump number, and it is independent of all other sources of randomness, cf.~\cite[Chapter~XI.2]{Asmussen_AppliedProbandQueues}. \\
Altogether, we may assume any MAP $(X,J)$ to be c\`adl\`ag. This allows to derive the path decomposition of its additive component
\begin{align}\label{MAPpathdescription}
	X_t=X_0+X_{1,t} +X_{2,t}=X_0+ \int_{(0,t]}  \dd  X_{s}^{(J_{s})}+\sum_{n\geq 1}\sum_{\substack{i,j\in S,\\i\neq j}} Z_{X,n}^{ij} \mathds{1}_{\lbrace J_{T_{n-1}}=i,J_{T_{n}}=j,T_n\leq t\rbrace}.
\end{align}
Conversely, assuming that  $J$ is a continuous-time Markov chain with state space $S$ and jump times $\{T_n, n\in\NN\}$, and $X$ has a path decomposition as in \eqref{MAPpathdescription}, the process $(X,J)$ is a MAP. \\
In this paper we always assume that $X_0=0$. \\
We denote the intensity matrix of the background driving chain $J$ by $\fatQ=(q_{ij})_{i,j\in S}$. 
We will assume $J$ to be ergodic with unique stationary distribution  $\fatpi=(\pi_j)_{j\in S}$ and write   
\begin{align*}
	\PP_\pi(\cdot):=\sum_{j\in S}\pi_j\PP_j(\cdot).
\end{align*}

We refer to \cite{Asmussen_AppliedProbandQueues, DDK_Annals} for more thorough introductions into the theory of MAPs.

\subsection{Markov modulated generalized Ornstein-Uhlenbeck processes}

Given a bivariate MAP $((\xi,\eta),J)=((\xi_t,\eta_t),J_t)_{t\geq 0}$, the \emph{Markov modulated generalized Ornstein-Uhlenbeck  (MMGOU) process driven by $((\xi,\eta),J)$} has been defined in \cite{BS_MMGOU} as the process $(V_t)_{t\geq 0}$ given by 
\begin{align}\label{MMGOUexplicit}
	V_t=e^{-\xi_t} \left(V_0+\int_{(0,t]} e^{\xi_{s-}} \dd\eta_s\right),
\end{align}
where the random variable $V_0$ is conditionally independent of $((\xi_t, \eta_t),J_t)_{t\geq 0}$ given $J_0$.
Moreover, it has been shown in \cite{BS_MMGOU} that this MMGOU process is the unique solution of the stochastic differential equation 
\begin{align}\label{MMGOUSDE}
	\dd V_t=V_{t-}\dd U_t+ \dd K_t,\qquad t\geq 0, 
\end{align}
for another bivariate Markov additive process $((U,K),J)=((U_t, K_t),J_t)_{t\geq 0}$ which is uniquely determined by $((\xi,\eta),J)$. \\
Further, in that source, assuming that the background driving Markov chain $J$ is ergodic with stationary distribution $\pi$, necessary and sufficient conditions for stationarity of the MMGOU process have been derived. Moreover, in \cite{BSdT}, moments of (stationary) MMGOU processes are studied.

\subsection{Notations}

Throughout our expositions we use the notations d-$\lim$ or $\overset{d}\rightarrow$ for distributional convergence of random variables. For a finite set $S$, $\cP(S)$ is the power set, while $\cB(\RR)$ and $\cB(\RR_+)$ denote the Borel $\sigma$ algebra on $\RR$ and $\RR_+=[0,\infty)$, respectively. \\
When considering stochastic processes in $\RR^2$ (or $\RR^d$), we will deliberately switch between the notation as column vector $\Big(\begin{smallmatrix}
	\xi \\ \eta 
\end{smallmatrix}\Big)$ or row vector $(\xi,\eta)$ without any indication to simplify the reading. In all other instances we a-priori assume a vector $\mathbf{v} \in\RR^N$ to be a column vector, and we denote its transpose as $\mathbf{v}^\top$. Special vectors that we will frequently use are $\bone=(1,\ldots, 1)^\top$, and $\mathbf{e}_j=(0,\ldots, 0, 1, 0, \ldots, 0)^\top$ with the single non-zero entry in the $j$'s component. For a vector $\mathbf{a}=(a_j)_{j\in S}$, $\diag (a_j, j\in S)=\diag(\mathbf{a})$ denotes the diagonal matrix with entries $a_j,j=1,\dots, N$. Further, ``$\circ$'' means elementwise multiplication of matrices, while standard matrix multiplication is denoted with the usual multiplication sign ``$\cdot$'' that is also used for scalars.

\section{A COGARCH model in a Markovian environment} \label{S-MSCOGARCH}
\setcounter{theorem}{0} 

\subsection*{Definition of a COGARCH process in a Markovian environment}

The central step in the original derivation of the COGARCH volatility process in \cite{KLM2004} is the observation, that the defining equations of the GARCH time series \eqref{eq-GARCH} can be solved recursively. Applying the same approach on the MSGARCH time series \eqref{eq-MSGARCH} yields the expressions
\begin{align*}
	V_n&=  \sum_{i=0}^{n-1} \beta(M_i) \prod_{k={i+1}}^{n-1} \left(\lambda (M_k) \epsilon_k^2 + \delta(M_k) \right) + V_0 \cdot \prod_{k=0}^{n-1} \left(\lambda (M_k) \epsilon_k^2 + \delta(M_k) \right)\\
	&= \left( \int_{(0,n)} \beta(M_{\lfloor s \rfloor})  \exp\Bigg( - \sum_{k=0}^{{\lfloor s \rfloor}} \log \left(\lambda (M_k) \epsilon_k^2 + \delta(M_k) \right) \Bigg) \dd s  + V_0 \right) \\
	& \quad \cdot \exp\left( \sum_{k=0}^{n-1} \log \left(\lambda (M_k) \epsilon_k^2 + \delta(M_k) \right) \right), \\
	G_n &= \epsilon_n \sqrt{V_n} = \left(\sum_{i=0}^n \epsilon_i - \sum_{i=0}^{n-1} \epsilon_i \right)\sqrt{V_n},
\end{align*}  
for the squared volatility and price time series. We now continue as in \cite{KLM2004} and embed the above model into continuous time by replacing the innovations $(\epsilon_n)_{n\in\NN_0}$ by the increments of a Lévy process $(L_t)_{t\geq 0}$. Additionally, it is natural to replace the appearing discrete Markov chain $(M_k)_{k\in\NN_0}$ by a continuous time Markov chain $(J_t)_{t\geq 0}$, where  $(L_t)_{t\geq 0}$ and $(J_t)_{t\geq 0}$ are supposed to be independent.\\

Thus, let $(L_t)_{t\geq 0}$ be a Lévy process with non-zero Lévy measure, and let $(J_t)_{t\geq 0}$ be an ergodic Markov chain on $S=\{1,\ldots, N\}$, independent of $(L_t)_{t\geq 0}$. We use the same parametrization as in \cite{KLM2004} and fix constants $\beta(j)>0$, $\lambda(j)\geq 0, 1>\delta(j)> 0$, for $j=1,\ldots, N$.
Define the auxiliary bivariate Lévy processes
\begin{align}
	\begin{pmatrix}
		\xi_{t}^{(j)}\\
		\eta_{t}^{(j)} 
	\end{pmatrix}:= \begin{pmatrix}
		- t \log(\delta(j)) - \sum_{0<s\leq t} \log \left( 1+ \frac{\lambda(j)}{\delta(j)}(\Delta L_s)^2\right)\\
	 t \beta(j)
	\end{pmatrix},\quad  t\geq 0, j=1,\ldots, N, \label{eq-LevyxietaMSCOGARCH}
\end{align}
 and define a process $((\xi,\eta),J)=((\xi_t,\eta_t),J_t)_{t\geq 0}$ by setting 
 \begin{align} \label{eq-xietaMSCOGARCHsimple} \begin{pmatrix}
 	\xi_t\\ \eta_t 
 \end{pmatrix} := \begin{pmatrix}
 	\int_{(0,t]} \dd \xi_s^{(J_s)} \\ \int_{(0,t]} \dd \eta_s^{(J_s)}
 \end{pmatrix}, \quad t\geq 0.\end{align}
A natural definition of a Markov switching COGARCH squared volatility process is given by 
$$V_t:= e^{-\xi_t} \Bigg(V_0 + \int_{(0,t]} e^{\xi_{s-}} \dd \eta_s \Bigg), \quad t\geq 0.$$

In order to enhance the proposed squared volatility process with the needed structure, we make the following observation.

\begin{lemma}\label{lem-xietaMAP}
Let $\FF=(\cF_t)_{t\geq 0}$ be the augmented natural filtration induced by $(L,J)$.	Then $((\xi,\eta),J)$ as defined in \eqref{eq-xietaMSCOGARCHsimple} is an $\FF$-MAP. 
\end{lemma}
\begin{proof}
Clearly $((\xi,\eta),J)$ is adapted to $\FF$ by construction. Further, let $0\leq s\leq t$, $B\in \cB(\RR^2)$, $C\in \cP(S)$, then
\begin{align*}
\lefteqn{\PP((\xi_t,\eta_t)\in B, J_t\in C | \cF_s)}\\
&= \PP\Bigg( \bigg( \xi_s - \int_{(s,t]} \log(\delta(J_u)) \dd u - \sum_{s<u\leq t} \log\Big(1+\frac{\lambda(J_u)}{\delta(J_u)} (\Delta L_u)^2\Big), \eta_s + \int_{(s,t]} \beta(J_u) \dd u \bigg) \in B, \\ & \qquad \qquad J_t\in C \bigg| \cF_s  \Bigg)\\
&= \PP\left(  (\xi_t,\eta_t)\in B, J_t\in C | (\xi_s, \eta_s), J_s  \right)
\end{align*}
due to the Markov property of $J$ as well as the independent increments of the Lévy process $L$. Thus $((\xi,\eta),J)$ is a Markov process.\\ 
It remains to prove the MAP property. Observe that the Lévy processes $\xi^{(j)}$ in \eqref{eq-LevyxietaMSCOGARCH} are dependent. Therefore we can not directly argue via the path decomposition \eqref{MAPpathdescription} at this point. Still, for any $s,t\geq 0$, and any $f:\RR^2\to \RR$, $g:S\to \RR$ bounded and measurable,
\begin{align*}
\lefteqn{	\EE_{J_0} [f(\xi_{s+t} - \xi_s, \eta_{s+t} - \eta_s) g(J_{s+t})) |\cF_s ] }\\
&= \EE_{J_0} \Bigg[  f\bigg( - \int_{(s,t]} \log(\delta(J_u)) \dd u - \sum_{s<u\leq t} \log\Big(1+\frac{\lambda(J_u)}{\delta(J_u)} (\Delta L_u)^2\Big), \int_{(s,t]} \beta(J_u) \dd u \bigg) g(J_{s+t}) \bigg| \cF_s  \Bigg]\\
		&= \EE_{J_s}[f(\xi_t, \eta_t) g(J_t)],
\end{align*}
due to the stationary increments of the Lévy process $L$, which proves the claim.
\end{proof}

 The path decomposition \eqref{MAPpathdescription} of the MAP $((\xi,\eta), J)$ motivates us to expand the definition of a Markov switching COGARCH model by allowing additional jumps at times of regime switches as indicated in the introduction. The resulting definition is as follows.

 \begin{definition}\rm
 	Let $L=(L_t)_{t\geq 0}$ be a Lévy process with non-zero Lévy measure $\nu_L$, and let $J=(J_t)_{t\geq 0}$ be an ergodic Markov chain on $S=\{1,\ldots, N\}$, independent of $L$. Let $\beta(j)>0$, $\lambda(j)\geq 0, 1>\delta(j)> 0$, for $j\in S$, be constants.
 	Define the auxiliary bivariate Lévy processes $(\xi_{t}^{(j)}, \eta_{t}^{(j)})_{t\geq 0}$, $j\in S$, via \eqref{eq-LevyxietaMSCOGARCH} and fix some jump distributions $F^{ij}$, $i,j\in S$, $i\neq j$, on $\RR\times \RR_+$.\\
 Define the MAP $((\xi,\eta),J)$ by setting 
 \begin{align}\begin{pmatrix}
 	\xi_t\\ \eta_t 
 \end{pmatrix} := \begin{pmatrix}
 \int_{(0,t]} \dd \xi_s^{(J_s)} \\ \int_{(0,t]} \dd \eta_s^{(J_s)}
\end{pmatrix} + \sum_{n\geq 1}\sum_{\substack{i,j\in S,\\i\neq j}}Z_{n}^{ij} \mathds{1}_{\lbrace J_{T_{n-1}}=i,J_{T_{n}}=j,T_n\leq t\rbrace}, \label{eq-xietaMMCOGARCH} \end{align} 
for i.i.d. sequences $\{Z_n^{ij}, n\in\NN\}$ with distribution $F^{ij}$, independent of all other sources of randomness. Then the \emph{Markov switching COGARCH (MSCOGARCH) model} $(V,G)=(V_t,G_t)_{t\geq 0}$  consists of squared volatility and price process given by
\begin{align} \label{eq-MSCOGARCH}
	V_t&:= e^{-\xi_t} \bigg(V_0 + \int_{(0,t]} e^{\xi_{s-}} \dd \eta_s \bigg),\\
		G_t&:= \int_{(0,t]} \sqrt{V_{t-}} \dd L_t, \quad t\geq 0, \nonumber 
\end{align} 
for some random variable $V_0$ that is conditionally independent of $((\xi,\eta),J)$ given $J_0$.
 \end{definition}

Observe that we chose the additional jumps at regime switches to take values in $\RR\times \RR_+$ in order to ensure positivity of the squared volatility process.

\begin{figure}[h]
	\includegraphics[width=\textwidth]{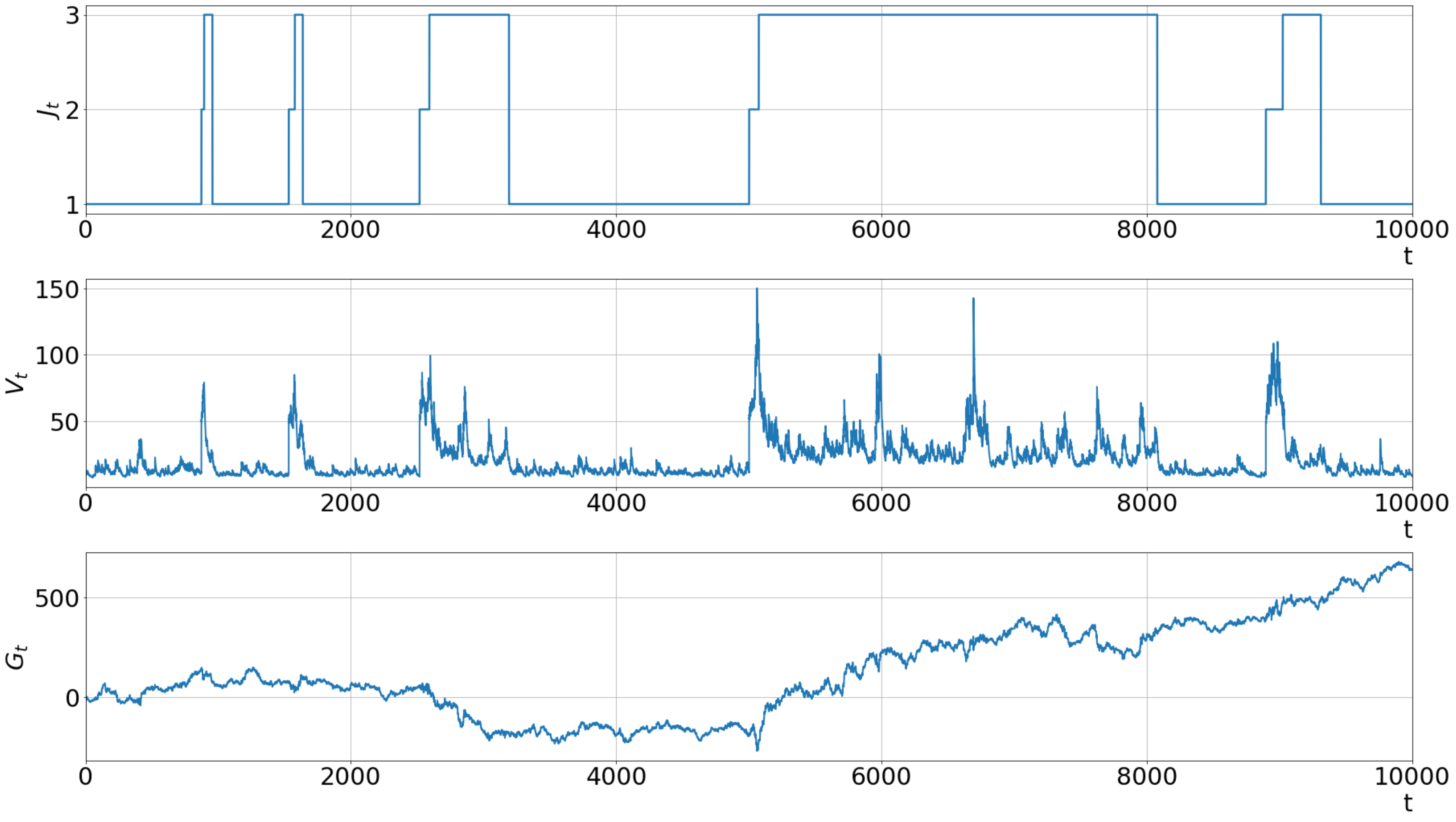}
	\caption{Simulated paths of a background driving Markov chain $(J_t)_{t\geq 0}$ with $|S|=3$ states (top), a resulting squared MSCOGARCH volatility process (middle), and a resulting MSCOGARCH price process (bottom). The simulated driving Lévy process $(L_t)_{t\geq 0}$ has been chosen as compound Poisson process with standard normal distributed jumps. The parameters of the three COGARCH regimes are $\delta=(0.9, 0.93, 0.92)$, $\lambda=(0.042, 0.047, 0.044)$, and $\beta=(0.7, 2, 1)$, such that all appearing COGARCH processes are stationary and the mean volatility is highest in regime $2$, lowest in regime 1. The additional jumps at times of regime switches have been realized as exponentially distributed random variables.}
\end{figure}

\begin{remark}\rm
	It follows from the above definition that between two consecutive jumps of the background driving chain $(J_t)_{t\geq 0}$ the MSCOGARCH process $(V,G)$ behaves just as a standard COGARCH process with parameters $\beta(J_t), \lambda(J_t), \delta(J_t)$ as defined in \eqref{eq-COGARCH}. In case the additional jumps at times of regime switches are all set to zero, the MSCOGARCH can thus be seen as a concatenation of COGARCH processes, see also \cite[Rem. 2.15]{BS_MMGOU}. As mentioned in the introduction such concatenations have been introduced e.g. in \cite{MbaMwambi} for a two-state background driving Markov chain, and lately in \cite{LiLiu}.
\end{remark}

By comparing the definition of the squared volatility in \eqref{eq-MSCOGARCH} and the process given in \eqref{MMGOUexplicit} we immediately note that the squared MSCOGARCH volatility is a special case of an MMGOU process. This allows us to apply various results obtained in \cite{BS_ExpFuncMAP, BS_MMGOU, BSdT} to derive basic properties of the squared MSCOGARCH volatility as presented subsequently.

As we will frequently need to split up the components of the additional jumps at times of a regime switch, in comparison to  \eqref{eq-xietaMMCOGARCH} we use the notations 
\begin{equation} \label{eq_addjumpscomponents} \sum_{n\geq 1}\sum_{\substack{i,j\in S,\\i\neq j}}Z_{n}^{ij} \mathds{1}_{\lbrace J_{T_{n-1}}=i,J_{T_{n}}=j,T_n\leq t\rbrace} =: \sum_{\substack{i,j\in S,\\i\neq j}}\sum_{0<s\leq t}\left( \begin{smallmatrix}
	\Delta \xi_{2,s}^{ij} \\ \Delta \eta_{2,s}^{ij}
\end{smallmatrix}\right)  =: \sum_{0<s\leq t}\left( \begin{smallmatrix}
	\Delta \xi_{2,s} \\ \Delta \eta_{2,s}
\end{smallmatrix}\right),\end{equation}
such that in particular 
\begin{align*}
	\begin{pmatrix}
		\xi_{t}\\
		\eta_{t} 
	\end{pmatrix}= \begin{pmatrix}
		- \int_{(0,t]} \log(\delta(J_s)) \dd s - \sum_{0<s\leq t} \log \left( 1+ \frac{\lambda(J_s)}{\delta(J_s)}(\Delta L_s)^2\right) + \sum_{0<s\leq t} \Delta \xi_{2,s}\\
		\int_{(0,t]} \beta(J_s) \dd s  + \sum_{0<s\leq t} \Delta \eta_{2,s}
	\end{pmatrix},\quad  t\geq 0.
\end{align*}

We start by providing the stochastic differential equation solved by the MSCOGARCH process.

\begin{proposition}\label{prop_MarkovMSCOGARCH}
The process $((V,G), J)$ is a Markov process and adapted to $\FF$, the augmented natural filtration induced by $(L,J)$. Moreover, the MSCOGARCH process $(V,G)$ satisfies the system of stochastic differential equations
\begin{align*}
	\dd V_t&= V_{t-} \dd U_t + \dd K_t,\\
	\dd G_t&= \sqrt{V_{t-}} \dd L_t,  \quad t\geq 0,
\end{align*}
with $G_0:=0$ and for the $\FF$-MAP $((U,K), J)$ with additive component
	\begin{align}
		\begin{pmatrix}
			U_{t}\\
			K_{t} 
		\end{pmatrix}:=  \begin{pmatrix}
			\int_{(0,t]} \log(\delta(J_s)) \dd s + \sum_{0<s\leq t}  \frac{\lambda(J_s)}{\delta(J_s)}(\Delta L_s)^2     + \sum_{0<s\leq t} \left(e^{ -\Delta \xi_{2,s}} -1 \right) \\
			\int_{(0,t]} \beta(J_s) \dd s  + \sum_{0<s\leq t}e^{ -\Delta \xi_{2,s}} \Delta \eta_{2,s}
		\end{pmatrix},\quad  t\geq 0.
		\label{eq-ULMSCOGARCH}
	\end{align}
\end{proposition}
\begin{proof}
	The given SDE for the squared volatility process, as well as the fact that $((U,K),J)$ is a MAP, follow from Lemma \ref{lem-xietaMAP} by an application of \cite[Prop. 2.11]{BS_MMGOU}. The SDE for the price process is immediate from the definition. Adaptedness of $((V,G),J)$ to $\FF$ is clear by construction. Lastly, to prove the Markov property of $((V,G),J)$, let $0\leq s\leq t$, $B\in \cB(\RR_+\times \RR)$, $C\in \cP(S)$, then
	\begin{align*}
		\lefteqn{\PP((V_t,G_t)\in B, J_t\in C | \cF_s)}\\
		&= \PP\Bigg( \bigg(  e^{-(\xi_t - \xi_s)} V_s + \int_{(s,t]} e^{-(\xi_t-\xi_{u-})} \dd \eta_u , G_s + \int_{(s,t]} \sqrt{V_{u-}} \dd L_u \bigg) \in B, J_t\in C \bigg| \cF_s  \Bigg)\\
		&= \PP\left(  (V_t,G_t)\in B, J_t\in C | (V_s, G_s), J_s  \right),
	\end{align*}
	by the MAP property of $((\xi,\eta),J)$ in the form as proven in \cite[Lemma 2.2]{BS_MMGOU}, and due to the independent increments of $L$.
\end{proof}

\subsection*{Stationarity of the MSCOGARCH volatility process}

As the volatility process in the COGARCH model is assumed to be stationary, we continue our studies with conditions for stationarity of the squared MSCOGARCH volatility process and provide a representation of its stationary distribution in Theorem \ref{thm-statMSCOGARCH}. Necessary and sufficient conditions for stationarity of an MMGOU process have been derived in \cite{BS_MMGOU}, and these may clearly be applied to derive necessary and sufficient conditions for stationarity of the squared MSCOGARCH volatility process. However, as the resulting conditions are technically difficult and hard to check in concrete settings, in this exposition we only present sufficient conditions for the existence of a stationary MSCOGARCH volatility that also allow for a comparison with the classical COGARCH case.

\begin{theorem}\label{thm-statMSCOGARCH}
Consider the squared MSCOGARCH volatility process $(V_t)_{t\geq 0}$ defined in \eqref{eq-MSCOGARCH} and assume that $\lambda(j)>0$ for at least one $j\in S$. Denote the  first return time of $J$ to $j$ by 
$$\tau_1(j):=\inf\{t> 0: J_t=j, J_{t-}\neq j \}.$$
If
\begin{align}
	\kappa_\xi&:= \sum_{j\in S} \pi_j \bigg(-\log(\delta(j)) - \int_\RR \log \Big(1+\frac{\lambda(j)}{\delta(j)} y^2 \Big) \nu_L(\dd y) + \sum_{k\in S\setminus\{j\}}  q_{jk} \int_{\RR\times \RR_+} x \dd F^{jk}(x,y) \bigg)>0, \nonumber \\
	&\text{and }\qquad \qquad  \int_{(1,\infty)} \log q \,\PP_j\bigg(\sup_{T_n\leq  \tau_1(j)} e^{-\xi_{2,T_n}} \Delta \eta_{2,T_n} \in \dd q\bigg)<\infty, \label{eq-statcondint}
\end{align}
and if one chooses
$$V_0\overset{d}= V_\infty:= \text{d-}\lim_{t\to\infty} e^{-\xi_t} \int_{(0,t]} e^{\xi_{s-}} \dd K_s,$$
conditionally independent of $L$ given $J_0$, then $(V_t)_{t\geq 0}$ is strictly stationary. In this case the corresponding price process $(G_t)_{t\geq 0}$ has stationary increments.
\end{theorem}
\begin{proof}
As shown in \cite[Thm. 3.3]{BS_MMGOU}, the MMGOU process $(V_t)_{t\geq 0}$ defined in \eqref{eq-MSCOGARCH} admits a strictly stationary solution if and only if either there exists a sequence $\{c_j,j\in S\}$ such that the resulting process is discrete with $V_t=c_{J_t}$ $\PP_\pi$-a.s. for all $t\geq 0$, or the exponential functional 
\begin{equation}\frace_{(-\xi^\ast, -K^\ast)}(t):= -\int_{(0,t]} e^{\xi^\ast_{s-}} \dd K^\ast_s \label{eq-fracedual} \end{equation}
converges in $\PP_\pi^\ast$-probability to some proper random variable as $t\to \infty$. Hereby, $((\xi^\ast,K^\ast), J^\ast)$ denotes the time-reversed MAP of $((\xi,K),J)$, i.e. a MAP such that for all $t\geq 0$
$$((\xi_{(t-s)-}-\xi_t, K_{(t-s)-}-K_t), J_{(t-s)-})_{0\leq s\leq t} \text{ under } \PP_\pi \text{ equals in law to } ((\xi_s^\ast, K_s^\ast), J_s^\ast)_{0\leq s\leq t} \text{ under } \PP_\pi^\ast,$$
and $\PP^\ast_j(\cdot):=\PP(\cdot|J_0^\ast=j)$;
see e.g. \cite[Appendix A.2]{DDK_Annals} for details.   \\
We start by showing that in our setting $(V_t)_{t\geq 0}$ can never be discrete: It follows from \cite[Prop.~4.7]{BS_ExpFuncMAP} (see also \cite[Eq. (3.8)]{BS_MMGOU} although that formula contains a wrong sign in front of the second integral) that the discrete solution $V_t=c_{J_t}$ is obtained if and only if $\PP_\pi$-a.s.
$$K_t=-\int_{(0,t]} c_{J_{s-}} \dd U_s + \int_{(0,t]} \dd c_{J_s}, \quad t\geq 0.$$
Inserting the given form of $((U,K),J)$ as presented in \eqref{eq-ULMSCOGARCH} this can be shown to yield $\PP_\pi$-a.s.
\begin{equation} \int_{(0,t]} \beta(J_s) \dd s + \int_{(0,t]} c_{J_s} \cdot \log(\delta(J_s)) \dd s + \sum_{0<s\leq t} c_{J_{s-}} \frac{\lambda(J_s)}{\delta(J_s)} (\Delta L_s)^2 = 0, \quad t\geq 0. \label{eq-conddiscrete} \end{equation}
Separating the continuous and the jump part of \eqref{eq-conddiscrete} we note that the continuous part vanishes if and only if $\beta(j)+ c_j \log(\delta(j))\equiv 0$. This implies $c_j=-\frac{\beta(j)}{\log(\delta(j))}>0$ for all $j$, and, as $\nu_L\not\equiv 0$,  a solution to \eqref{eq-conddiscrete} can only exist if  $\lambda(j)\equiv0$ which has been ruled out be assumption. Thus $(V_t)_{t\geq 0}$ can not admit a stationary discrete solution. \\
With this, as $S$ is chosen to be finite, we conclude by \cite[Rem. 4.2]{BS_ExpFuncMAP} that convergence of the functional \eqref{eq-fracedual} in $\PP_\pi^\ast$-probability is equivalent to $\PP_\pi^\ast$-a.s. convergence. Hence, summarizing the above, $(V_t)_{t\geq 0}$ admits a strictly stationary solution if and only if the functional \eqref{eq-fracedual} converges $\PP_\pi^\ast$-a.s. as $t\to\infty$, in which case its limit under $\PP_\pi^\ast$ is equal in law to the distributional limit $V_\infty$ as $t\to \infty$ of 
$$\mathfrak{F}_{(\xi,K)}(t) := e^{-\xi_t} \int_{(0,t]} e^{\xi_{s-}} dK_s$$
under $\PP_\pi$, see \cite[Lemma 3.1]{BS_ExpFuncMAP}. Further, as shown in \cite[Thm. 3.3]{BS_MMGOU}, in this case the stationary law of $(V_t)_{t\geq 0}$ is given by the law of $V_\infty$. \\
It remains to be checked that under our conditions the functional \eqref{eq-fracedual} converges $\PP_\pi$-a.s. as $t\to\infty$. By a combination of \cite[Props. 5.2 and 5.7.1]{BS_ExpFuncMAP} this convergence follows, if the long term mean $\kappa_{-\xi^\ast}$ of $-\xi^\ast$ is positive, and moreover, as $\eta$ and hence $K$ only jump at times of regime switches,
\begin{equation} \int_{(1,\infty)} \frac{\log q}{\bar{A}_{-\xi^\ast}(\log q)} \PP_j^\ast\Big(\sup_{0\leq T_n \leq \tau_1^\ast(j)} e^{\xi^\ast_{T_n-}} |\Delta K_{T_n}^\ast | \in \dd q\Big)<\infty, \label{eq-statcondhelp1}\end{equation}
with 
\begin{align}
	\bar{A}_{-\xi^\ast} (x)&:= \sum_{j\in S} \pi_j^\ast \bigg(\gamma_{-\xi^{(j),\ast}} + \nu_{-\xi^{(j),\ast}}((1,\infty)) + \int_1^x \nu_{-\xi^{(j),\ast}}((y,\infty)) \dd y + \sum_{i\in S\setminus\{j\}} q_{ij}^\ast \EE[(-\Delta \xi^{ij,\ast}_2)^+] \bigg) \label{eq-defAbar}\\
	&=  \sum_{j\in S} \pi_j \bigg(-\log(\delta(j)) + \sum_{i\in S\setminus\{j\}} q_{ji} \EE[(\Delta \xi^{ji}_2)^+] \bigg)>0 \nonumber
\end{align}
since $-\xi^{(j),\ast}$ is equal in law to $\xi^{(j)}$ which has drift $-\log(\delta(j))$ and no positive jumps by construction. In particular $\bar{A}_{-\xi^\ast} (x)$ does not depend on $x$, and therefore \eqref{eq-statcondhelp1} follows if
\begin{align*}&  \int_{(1,\infty)} \log q \,\PP_j^\ast\Big(\sup_{0\leq T_n \leq \tau_1^\ast(j)} e^{\xi^\ast_{T_n-}} |\Delta K_{T_n}^\ast | \in \dd q\Big)<\infty \\
\Leftrightarrow  &\int_{(1,\infty)} \log q \,\PP_j \Big(\sup_{0\leq T_n \leq \tau_1(j)} e^{-\xi_{T_n-}} |\Delta K_{T_n}| \in \dd q\Big)<\infty \end{align*}
and as $\Delta K_t = e^{-\Delta \xi_2,t} \Delta \eta_{2,t}$ the latter is equivalent to the claimed integral condition \eqref{eq-statcondint}.\\
To check positivity of the long term mean $\kappa_{-\xi^\ast}$ note that by its definition, cf. \cite[Eq. (3.6)]{BS_ExpFuncMAP}, 
\begin{align*}
	\kappa_{-\xi^\ast} &= \sum_{j\in S} \pi_j^\ast \bigg(\EE[-\xi^{(j),\ast}_1] + \sum_{i\in S\setminus\{j\}} q_{ji}^\ast \EE[-\Delta \xi^{ji, \ast}_2] \bigg) = \sum_{j\in S} \pi_j \bigg( \EE[\xi^{(j)}_1]  + \sum_{i\in S\setminus\{j\}} q_{ji} \EE[\Delta \xi^{ji}_2] \bigg) = \kappa_\xi,
\end{align*}
due to the given form of $\xi$. Thus positivity of $\kappa_{-\xi^\ast}$ follows by assumption.\\
Lastly, assuming strict stationarity of $V$, stationarity of the increments of $G$ is now an immediate consequence of the stationary increments of $L$.
\end{proof}

\begin{remark}\rm \label{rem-statCOGARCH}
		In the special case $|S|=1$ the condition $\kappa_\xi>0$ in Theorem \ref{thm-statMSCOGARCH} reduces to  
	\begin{equation} \label{eq-COGARCHstat} \int_\RR \log\Big(1+\frac{\lambda}{\delta} y^2 \Big)\nu_L(\dd y) <-\log \delta,\end{equation}
	which is just the condition that has been derived in  \cite[Thms. 3.1 and 3.2]{KLM2004} as necessary and sufficient for strict stationarity of the classical COGARCH volatility \eqref{eq-COGARCH}.\\
	For $|S|>1$ we note that $\kappa_\xi>0$ can be fulfilled even if for some regime states \eqref{eq-COGARCHstat} does not hold, i.e. the regime switching behavior can balance out short times of non-stationarity. Moreover, in the presence of additional jumps at times of regime switches, large jumps in the $\xi$-component can even improve stationarity as they increase $\kappa_\xi$. However, the presence of such shock jumps induces a possible dependence between $\xi$ and $\eta$, which then leads to the necessity of an additional integral condition as the one stated in \eqref{eq-statcondint}. This coincides with the well-known behavior of exponential functionals driven by bivariate Lévy processes as studied in \cite{EricksonMaller}. 
	\end{remark} 

\begin{remark}\rm
	If $\lambda(j)\equiv 0$ the MSCOGARCH volatility loses its dependence on the noise process $(L_t)_{t\geq 0}$, and consequently it is presumably not relevant for practical applications. In particular, $\lambda(j)\equiv 0$ implies that the driving process $(\xi,\eta)$ given $J$, and hence also the squared volatility process, is deterministic except of the jumps at regime switches.\\
	Still,  by an application of \cite[Thm. 3.3]{BS_MMGOU} and direct computations based on \cite[Eq. (3.8)]{BS_MMGOU}, under this condition, the squared MSCOGARCH volatility process admits a stationary solution if and only if the jumps $\left( 
		\Delta \xi_{2,s}^{ij} , \Delta \eta_{2,s}^{ij}  \right)$ at times of regime switches fulfill
		$$e^{-\Delta \xi_{2,\cdot}^{ij}} \left(\Delta \eta_{2,\cdot}^{ij} - \frac{\beta(i)}{\log \delta(i)} \right) + \frac{\beta(j)}{\log\delta(j)} = 0,$$
		and hence either of them is a function of the other. In this case the stationary distribution is discrete and the squared volatility process is given by 
		$$V_t= -\frac{\beta(J_t)}{\log \delta(J_t)}, \quad t\geq 0,$$
		which is a continuous-time Markov chain whenever the values $\frac{\beta(j)}{\log \delta(j)}$ for $j\in \mathcal{S}$ are pairwise different.
\end{remark}

\subsection*{Moments and autocorrelation structure}

In this section we study moments of the stationary squared MSCOGARCH volatility process and the corresponding price process. In order to formulate our results we need to introduce the \emph{matrix exponent} $\mathbf{\Psi}_\xi$ of the MAP $(\xi,J)$ given by 
\begin{equation}\label{eq_matrixexp}
	\mathbf{\Psi}_\xi(w)=\diag \left(\psi_j(w), j\in S\right)+\fatQ^\top\circ \left(\EE\left[e^{w \Delta \xi_2^{ij}}\right]\right)_{ij\in S}^\top,  
\end{equation}
for all $w\in\CC$ such that the right hand side exists, see also \cite[Prop. XI.2.2]{Asmussen_AppliedProbandQueues} or \cite{DDK_Annals}. 
Hereby and in the following, $\psi_j(w)=\log \EE[e^{w \xi^{(j)}_1}]$ is the Laplace exponent corresponding to the Lévy process $\xi^{(j)}$, which in the present situation reads as
\begin{equation}
	\label{eq_psixij}
	\psi_j(w) = -w \log \delta(j) + \int_\RR \Big(\Big(1+\frac{\lambda(j)}{\delta(j)} y^2 \Big)^{-w} - 1\Big)\nu_L(\dd y).
\end{equation} 
The following lemma provides a representation of certain values of the matrix exponent in terms of the model parameters. To shorten our notation in the upcoming results, we define the $N\times N$ matrices \begin{equation}\label{eq-defFmatrix} \mathbf{F}_{k,n}:= \left( \int_{\RR\times \RR_+} e^{-kx} y^n \dd F^{ij}(x,y) \right)_{i,j\in S}, \end{equation}
for all $k,n\in\NN_0$ such that the appearing integrals are finite.

\begin{lemma}
	For any $k\in \NN$ such that there exists $t>0$ with $\EE_\pi[\sup_{0<s\leq t} e^{-k\xi_s}]<\infty$, in the MSCOGARCH model it holds
	\begin{align*}
		\mathbf{\Psi}_\xi(-k) &= \diag\left(k\log (\delta(j)) + \int_\RR\Big( \Big(1+\frac{\lambda(j)}{\delta(j)} y^2\Big)^k - 1 \Big)\nu_L(\dd y), j\in S \right)  + \fatQ^\top \circ \mathbf{F}_{k,0}^\top.
	\end{align*}
\end{lemma}
\begin{proof}
	This follows immediately from \eqref{eq_matrixexp} and \eqref{eq_psixij}.
\end{proof}

\begin{proposition}\label{prop_meanvarianceMSCOGARCH}
	Assume that for all $i,j\in S$
\begin{equation}\label{eq-condmoments} \EE_j[|V_0|^k]<\infty,\, \EE[e^{k |\xi_1^{(j)}|}]<\infty, \int_{\RR\times \RR_+} e^{k |x|} \dd F^{ij}(x,y)<\infty ,  \text{ and } \int_{\RR\times \RR_+} |y|^{k} \dd F^{ij}(x,y) <\infty,\end{equation} 
for $k=1$, then for all $t\geq 0$ and $j\in S$
	\begin{align*}
		\EE_j[V_t]&= \bone^\top e^{\mathbf{\Psi}_\xi(-1) \cdot t}\, \mathbf{e}_j \cdot \EE_j[V_0] + \bone^\top \int_0^t e^{\mathbf{\Psi}_\xi(-1) \cdot (t-s)} \left( \diag(\mathbf{\beta}) + \fatQ^\top \circ \mathbf{F}_{1,1}^\top \right)e^{\fatQ^\top \cdot s} \mathbf{e}_j \dd s.
	\end{align*}
Furthermore, if \eqref{eq-condmoments} holds for $k=2$, then for all $0\leq s\leq t$ and $j\in S$
\begin{align*}
	\lefteqn{\cov_j(\mathbf{e}_{J_t} V_t, V_s)}\\ &= e^{\mathbf{\Psi}_\xi(-1) (t-s)} \cdot \cov(\mathbf{e}_{J_s} V_s, V_s)\\
	&\quad +  \left(\sum_{k=0}^\infty \frac{(t-s)^k}{k!} \sum_{i=1}^k (\mathbf{\Psi}_\xi(-1))^{i-1} \left( \diag(\mathbf{\beta}) + \fatQ^\top \circ \mathbf{F}_{1,1}^\top \right)(\fatQ^\top)^{k-i} \right) \cdot \cov_j(\mathbf{e}_{J_s}, V_s), \quad \text{with}\\
		\lefteqn{\cov_j(\mathbf{e}_{J_t}, V_s)} \\
		&= e^{\fatQ^\top (t-s)} \cov_j(\mathbf{e}_{J_s}, V_s),
\end{align*}
and  
$\cov_j(V_t, V_s) = \bone^\top \cov_j(e_{J_t} V_t, V_s) $
is decreasing exponentially if the maximal eigenvalue of $\mathbf{\Psi}_\xi(-1)$ is negative.
\end{proposition} 
\begin{proof}
	By \cite[Lemma 3.2]{BSdT} it follows from \eqref{eq-condmoments} that $\EE_\pi[\sup_{0<s\leq t} e^{k |\xi_s|}]<~\infty$. We may therefore apply \cite[Thms. 4.4 and 4.5]{BSdT} to derive the expectation and autocovariance function of the squared MSCOGARCH volatility process. The result now follows by inserting the special structure of the MAP $((U,K),J)$ and direct computation.
\end{proof}

As we are particularly interested in the stationary version of the MSCOGARCH volatility process, we also provide a recursion formula for the integer  moments of the stationary squared MSCOGARCH volatility.

\begin{proposition}\label{prop_momentsstat}
	Assume that for all $j\in S$ and a given $k\in\NN$
	\begin{enumerate}
		\item[(a)] $\psi_j(-k)<|q_{jj}|$ and
		\begin{align*}
			\max_{i\in S\setminus\{j\}}& \Big((|q_{ii}|-\psi_i(-k))^{-1} \sum_{\ell \in S\setminus\{i,j\}} q_{i\ell} \int_{\RR\times \RR_+} e^{-kx} \dd F^{i\ell}(x,y)\Big) <1, 
		\end{align*}
		\item[(b)] the maximal eigenvalue of $\mathbf{\Psi}_\xi(-k)$ is negative,
		\item[(c)] $$\int_{\RR\times \RR_+} e^{-k|x|} y^k \dd F^{ij}(x,y)<\infty \quad \text{for all }i,j\in S.$$ 
	\end{enumerate}
	Then the squared MSCOGARCH volatility process $(V_t)_{t\geq 0}$ has a stationary distribution $V_\infty$, and the $k$'th moment of  $V_\infty$ is given recursively as
	\begin{align*}
		\EE_\pi[V_\infty^k] = - \fatone^\top \cdot \mathbf{\Psi}_\xi(-k)^{-1}  \left( k \diag(\mathbf{\beta})\cdot  \fatpi \, \EE_\pi[V_\infty^{k-1}] + \sum_{n=1}^k {k\choose n} \Big(\fatQ^\top \circ \mathbf{F}_{k,n}^\top \Big) \cdot \fatpi \,\EE_\pi [V_\infty^{k-n}]\right),
	\end{align*}
	for $k\geq 1$ with starting value $\EE_\pi[V_\infty^0]=1.$
\end{proposition}
\begin{proof}
	This can be derived by direct computations from \cite[Thm. 4.8 and Rem. 4.9]{BSdT} taking the special structure of the MAPs \eqref{eq-xietaMMCOGARCH} and \eqref{eq-ULMSCOGARCH} into account. Note in particular that $U$ and $K$ have no Gaussian components and hence all quadratic variations of continuous parts of $U$ and $K$ in \cite[Rem. 4.9]{BSdT} vanish from the computations. Moreover, $U$ and $K$ can only jump simultaneously at time of regime switches which further simplifies the resulting formulas.
\end{proof}

\begin{remark}\rm 
	The conditions in Proposition \ref{prop_momentsstat} (a) are primarily needed to ensure finiteness of exponential moments of $\xi$. As one can see, the conditions are stronger in regime states that tend to be visited for long times, i.e. in states with small exit rates, while short term stays in other states can be balanced out.
\end{remark}

Lastly, we consider the increments of the price process 
$$G_t^{(r)}:= G_{t+r}-G_t$$
corresponding to log-returns of time periods of length $r>0$. As shown in the next proposition, these log-returns are uncorrelated on disjoint time intervals, but squared log-returns are in general correlated. This agrees with empirical findings as well as with the likewise behavior of the COGARCH model, cf. \cite[Prop. 5.1]{KLM2004}, and of discrete-time GARCH models, cf. \cite{FZ2}.

\begin{proposition}\label{prop_momentslogreturns}
 Assume that $L$ is a pure-jump Lévy process with $\EE[L_1]=0$ and $\EE[L_1^2]<\infty$, and that the conditions of Proposition \ref{prop_momentsstat} are fulfilled for $k=1$. Then for any $t\geq 0$ and $h\geq r >0$ it holds
 \begin{align*}
 	\EE_\pi \big[G_t^{(r)}\big]&=0,\\
 	\EE_\pi\Big[(G_t^{(r)})^2\Big] &= - \EE[L_1^2] \cdot r\cdot  \fatone^\top \mathbf{\Psi}_\xi(-1)^{-1}  \left( \diag(\mathbf{\beta})+ \Big(\fatQ^\top \circ \mathbf{F}_{1,1}^\top \Big)\right) \cdot \fatpi, \quad \text{and} \\
 	\cov_\pi\big(G_t^{(r)}, G_{t+h}^{(r)}\big) &= 0.
 \end{align*}
Further, assuming $\EE[L_1^4]<\infty$ and under the conditions of Proposition \ref{prop_momentsstat} for $k=2$, 
\begin{align*}
	\lefteqn{\cov_\pi\Big((G_t^{(r)})^2, (G_{t+h}^{(r)})^2 \Big)}\\ 
	 &=\EE[L_1^2]  \fatone^\top \mathbf{\Psi}_\xi(-1)^{-1} \Big( e^{h \mathbf{\Psi}_\xi(-1)} - e^{(h-r) \mathbf{\Psi}_\xi(-1)}\Big) \EE_\pi[G_r^2 V_r \mathbf{e}_{J_r}] \\
	& \quad + \EE[L_1^2] \cdot \fatone^\top \mathbf{\Psi}_\xi(-1)^{-1} \\
	& \qquad \cdot  \int_0^h \Big(e^{\mathbf{\Psi}_\xi(-1)(h-u)}- e^{\mathbf{\Psi}_\xi(-1)((h-u-r)\vee 0)}\Big) \Big(\diag(\mathbf{\beta}) + \fatQ^\top \circ \mathbf{F}_{1,1}^\top\Big) e^{\fatQ^\top u} \dd u\cdot 
	\EE_\pi[G_r^2 \mathbf{e}_{J_r}]\\
	&\quad  - \EE[L_1^2]^2\cdot r^2 \Big(\fatone^\top \mathbf{\Psi}_\xi(-1)^{-1} \Big( \diag(\mathbf{\beta}) + \Big(\fatQ^\top \circ \mathbf{F}_{1,1}^\top \Big)\Big) \fatpi \Big)^2.
\end{align*}
\end{proposition}
\begin{proof}
	As the price process $(G_t)_{t\geq 0}$ is defined as an integral with respect to the Lévy process $(L_t)_{t\geq 0}$, the given values for $\EE_\pi [G_t^{(r)}]$, $\EE_\pi [(G_t^{(r)})^2 ]$, and $\cov_\pi(G_t^{(r)}, G_{t+h}^{(r)} )$ can be obtained in complete analogy to the computations in the proof of \cite[Prop. 5.1]{KLM2004}, using Proposition \ref{prop_momentsstat}.\\
	For the covariance of the squared increments, let $\bar{\FF}=(\bar{\cF}_t)_{t\geq 0}$ denote the natural filtration induced by $(L,J,\xi_2, \eta_2)$. Then, again in analogy to \cite[Proof of Prop. 5.1]{KLM2004}, conditioning yields
	\begin{align*}
		\cov_\pi\Big((G_t^{(r)})^2, (G_{t+h}^{(r)})^2 \Big) &= \EE_\pi \Big[G_r^2 \, \EE\big[(G_h^{(r)})^2 |\bar{\cF}_r \big]\Big] - \Big(\EE[L_1]^2 \EE_\pi[V_\infty] r \Big)^2\\
		&= \EE[L_1^2] \EE_\pi\Big[G_r^2 \cdot \int_{(h,h+r]} \EE[V_s|\bar{\cF}_r] \dd s\Big] - \Big(\EE[L_1^2]^2 \EE_\pi[V_\infty] r \Big)^2,
	\end{align*}
for $h\geq r\geq 0$, $t\geq 0$. Hereby
\begin{align*}
	\EE[V_s|\bar{\cF}_r] &= \EE_{J_r}[V_{s-r}] + \big(V_r-\EE[V_0]\big) \EE_{J_r}\big[e^{-\xi_{s-r}}|\bar{\cF}_0\big]\\
	&= \fatone^\top e^{\mathbf{\Psi}_\xi(-1) (s-r)} \mathbf{e}_{J_r} V_r + \fatone^\top \int_0^{s-r} e^{\mathbf{\Psi}_\xi(-1)(s-r-u)} \Big(\diag(\mathbf{\beta}) + \Big(\fatQ^\top \circ \mathbf{F}_{1,1}^\top \Big)\Big) e^{\fatQ^\top u} \mathbf{e}_{J_r} \dd u,
\end{align*}
where the first equality follows by computations similar to \cite[Proof of Prop. 5.1]{KLM2004}, while the second equality is derived from Prop. \ref{prop_momentsstat} and \cite[Thm. 4.8]{BSdT}. The stated formula now  follows by direct computation.
\end{proof}

\subsection*{A special case with no jumps in the $\eta$-component}

To simplify the structure of the MSCOGARCH process, we consider a squared MSCOGARCH volatility process under the assumption of no additional jumps in the second component $\eta$ at times of a regime switch. Hence we set $\Delta \eta_{2, \cdot} \equiv 0$ which implies that the two components  of the driving process $(\xi,\eta)$ are conditionally independent given $J$. This leads to significant simplifications of the above results:
\begin{itemize}
	\item The integral condition \eqref{eq-statcondint} is always fulfilled and therefore, by Theorem \ref{thm-statMSCOGARCH} (as long as $\lambda(j)\not\equiv 0$), $\kappa_\xi>0$ is a sufficient condition for stationarity of the squared MSCOGARCH volatility, and thus for stationary increments of the price process. In particular, according to Remark \ref{rem-statCOGARCH}, stationarity of the single COGARCH volatility regimes is sufficient, but not necessary, for stationarity of the MSCOGARCH volatility.
	\item The moment formulas in Propositions \ref{prop_meanvarianceMSCOGARCH} and \ref{prop_momentslogreturns} simplify substantially, as $\mathbf{F}_{k,n}= \mathbf{0}$ for all $n\neq 0$. In particular the recursion in Proposition \ref{prop_momentsstat} can be solved explicitely giving 
	\begin{align*}
		\EE_\pi[V_\infty^k] 
		&= k! \prod_{n=1}^k \left(- \fatone^\top \cdot \mathbf{\Psi}(-n)^{-1}  \cdot \diag(\mathbf{\beta})\cdot  \fatpi \right).
	\end{align*}
	This nicely generalizes the product formula for the moments of the squared COGARCH volatility as it has been obtained in \cite[Prop. 4.2]{KLM2004}.
\end{itemize}

A drawback of this much simpler model however lies in the fact that jumps of the squared volatility at times of a regime switch are necessarily scaled by the previous value of the volatility. More precisely, jumps at times of regime switches $\{T_n, n\in\NN\}$ are always of the form 
$$\Delta V_{T_n}= V_{T_n-} \Delta U_{T_n} = V_{T_n-} \left(e^{ -\Delta \xi_{2,{T_n}}} -1 \right).$$
This may be too restrictive in order to capture truly exogenous shocks in the model.

\section{A Barndorff-Nielsen Shephard model in a Markovian environment} \label{S-MSBNS}
\setcounter{theorem}{0} 

\subsection*{Definition of a BNS model in a Markovian environment}

Following \cite{BarndorffNielsen-Shephard}, the squared volatility process $V$ in the BNS model as given in \eqref{eq-BNS} is parametrized by the parameter $\lambda$ that describes the dynamic structure of the process, and the driving Lévy process $L$ that determines the stationary distribution of the volatility process. Hereby, $L$ is assumed to be a pure-jump subordinator.\\
Thus, in order to define a BNS volatility in a Markovian environment, let $(J_t)_{t\geq 0}$ be an ergodic Markov chain on $S=\{1,\ldots, N\}$, fix constants $\lambda(j)>0$, and let $(L_t^{(j)})_{t\geq 0}$ for $j\in S$ be independent pure-jump subordinators, independent of the Markov chain $J$. Further, define the auxiliary bivariate Lévy processes
\begin{equation}
	\label{eq_BNSxietaLevy}
	\begin{pmatrix}
		\xi_t^{(j)}\\ \eta_t^{(j)} 
	\end{pmatrix} := \begin{pmatrix}
	t \lambda(j) \\ L_{t \lambda(j)}^{(j)}
\end{pmatrix}, \quad t\geq 0, j=1,\ldots, N.
\end{equation}
By construction the processes $\{(\xi^{(j)}, \eta^{(j)}), j=1,\ldots, N\}$ are independent Lévy processes, and hence it follows from the path decomposition \eqref{MAPpathdescription} that 
\begin{equation} \label{eq-xietaBNSsimple} \Big( \Big(\int_{(0,t]} \dd \xi_s^{(J_s)}, \int_{(0,t]} \dd \eta^{(J_s)}_s\Big), J_t\Big)_{t\geq 0} \end{equation} 
is a MAP with respect to its augmented natural filtration. We could therefore consider the MMGOU process driven by \eqref{eq-xietaBNSsimple} as squared volatility process of a Markov-switching BNS model. However, as in the MSCOGARCH model, we aim to incorporate jumps in the volatility at times of a regime switch. This then leads to the following definition of a BNS model in a Markovian environment.

\begin{definition}\rm \label{defMSBNS}
Let $W=(W_t)_{t\geq 0}$ be a standard Brownian motion and let $J=(J_t)_{t\geq 0}$ be an ergodic Markov chain on $S=\{1,\ldots, N\}$ with stationary distribution $\fatpi$, independent of $W$. Let $\lambda(j)> 0, \mu(j), \beta(j),  \rho(j)\in\RR$, for $j\in S$, be constants. Let $(L_t^{(j)})_{t\geq 0}$ for $j\in S$ be independent pure-jump subordinators with Lévy measures $\nu_{L^{(j)}}$, independent of $J$ and $W$.\\
	Define the auxiliary bivariate Lévy processes $(\xi_{t}^{(j)}, \eta_{t}^{(j)})_{t\geq 0}$, $j\in S$, via \eqref{eq_BNSxietaLevy} and fix some jump distributions $F^{ij}$, $i,j\in S$, $i\neq j$, on $\RR_+^2$.\\
	Define the MAP $((\xi,\eta),J)$ by setting 
	\begin{align}\begin{pmatrix}
			\xi_t\\ \eta_t 
		\end{pmatrix} := \begin{pmatrix}
			\int_{(0,t]} \dd \xi_s^{(J_s)} \\ \int_{(0,t]} \dd \eta_s^{(J_s)}
		\end{pmatrix} + \sum_{n\geq 1}\sum_{\substack{i,j\in S,\\i\neq j}}Z_{n}^{ij} \mathds{1}_{\lbrace J_{T_{n-1}}=i,J_{T_{n}}=j,T_n\leq t\rbrace}, \label{eq-xietaMMBNS} \end{align} 
	for i.i.d. sequences $\{Z_n^{ij}, n\in\NN\}$ with distribution $F^{ij}$, independent of all other sources of randomness.
	For $t\geq 0$ set 
	\begin{equation} \label{eq-deftildeeta}  \tilde{\eta}_t := \eta_t - \fatone^\top \bigg(\diag\Big(\lambda(j) \EE\big[L_1^{(j)}\big], j\in S\Big) + \fatQ^\top \circ \Big(\int_{\RR_+\times \RR_+} y \dd F^{i,j}(x,y)\Big)_{i,j\in S}^\top \bigg) \int_0^t \mathbf{e}_{J_s} \dd s.\end{equation} 
 Then the \emph{Markov switching BNS (MSBNS) model} $(V,G)=(V_t,G_t)_{t\geq 0}$  consists of squared volatility and price process given by
	\begin{align} \label{eq-MSBNS}
		V_t&:= e^{-\xi_t} \bigg(V_0 + \int_{(0,t]} e^{\xi_{s-}} \dd \eta_s \bigg)\\
		G_t&:= \int_{(0,t]} (\mu(J_s) + \beta(J_s) V_s) \dd s + \int_{(0,t]} \sqrt{V_s} \dd W_s + \int_{(0,t]} \rho(J_{s-}) \dd \tilde{\eta}_s, \quad t\geq 0, \nonumber 
	\end{align} 
	for some random variable $V_0$ that is conditionally independent of $((\xi,\eta),J)$ given $J_0$.
\end{definition}

\begin{remark}\rm
	The process $\tilde{\eta}$ as defined in \eqref{eq-deftildeeta} is a martingale with respect to $\FF$ under any initial distribution of $J$, see \cite[Thm. 3.8]{BSdT}.
\end{remark}

\begin{figure}[h]
	\includegraphics[width=\textwidth]{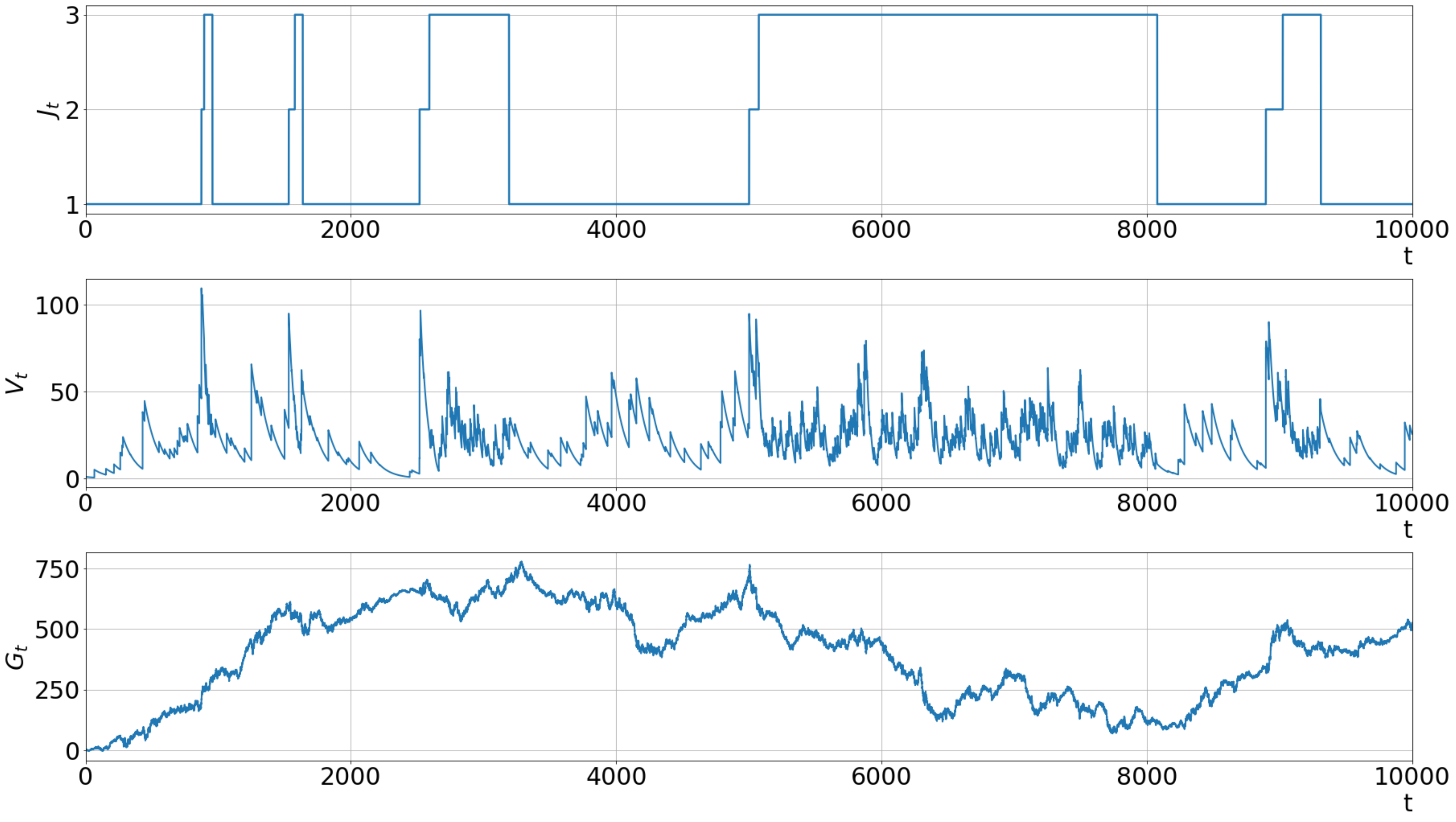}
	\caption{Simulated paths of a background driving Markov chain $(J_t)_{t\geq 0}$ with $|S|=3$ states (top), a resulting squared MSBNS volatility process (middle), and a resulting MSBNS price process driven by a standard Brownian motion (bottom). The simulated driving Lévy processes $(L_t^{(j)})_{t\geq 0}$ have been chosen as compound Poisson processes with intensities $(2,2,5)$ and exponentially distributed jumps with parameters $(0.1,0.1,0.2)$. The model parameters of the three MSBNS regimes are $\lambda=(0.01, 0.02, 0.04)$, $\mu = (0.1,0,0)$, and $\beta\equiv \rho \equiv 0$, such that all appearing BNS volatility processes are stationary and the volatility exhibits most jumps in regime $2$, fewest jumps in regime 1. The additional jumps at times of regime switches have been realized as exponentially distributed random variables.}
\end{figure}

In analogy to Proposition \ref{prop_MarkovMSCOGARCH} and using the notation for the jumps at regime switches as introduced in \eqref{eq_addjumpscomponents} we observe the following.

\begin{lemma} \label{lem_BNSSDE}
	The process $((V,G), J)$ defined in \eqref{eq-MSBNS} is a Markov process and adapted to $\FF$, the augmented natural filtration induced by $\{(L_{\lambda(j) t}^{(j)}), j\in S,J_t)$. Moreover, the MSBNS process $(V,G)$ satisfies the system of stochastic differential equations
	\begin{align*}
		\dd V_t&= V_{t-} \dd U_t + \dd K_t,\\
		\dd G_t&= (\mu(J_t) + \beta(J_t) V_t) \dd t + \sqrt{V_t} \dd W_t +\rho(J_{t-}) \dd \tilde{\eta}_t,  \quad t\geq 0,
	\end{align*}
	with $G_0:=0$ and for the $\FF$-MAP $((U,K), J)$ with additive component
		\begin{equation}
		\label{eq_BNSUK}
		\begin{pmatrix}
			U_t \\ K_t
		\end{pmatrix} := \begin{pmatrix}
			- \int_{(0,t]} \lambda(J_s) \dd s +\sum_{0<s\leq t} (e^{-\Delta \xi_{2,s}}-1) \\
			\int_{(0,t]} \dd \eta_{s}^{(J_s)} + \sum_{0<s\leq t} e^{-\Delta \xi_{2,s}} \Delta \eta_{2,s} 
		\end{pmatrix}, \quad t\geq 0.
	\end{equation}
\end{lemma}

\begin{remark}\rm \label{rem-defMSBNS}	
	In the MSBNS model additional jumps of the driving processes at times of regime switches are only allowed to take values in $\RR_+^2$. While the restriction to $\RR_+$ in the second component ensures positivity of the squared volatility process (as in the MSCOGARCH model), the restriction in the first component ensures the squared volatility to always \emph{decrease} exponentially towards its mean, hence to avoid explosion. \\
	From the above Lemma \ref{lem_BNSSDE} we additionally see that $$\Delta V_{T_n} = e^{-\Delta \xi_{T_n}}(V_{T_n-} + \Delta \eta_{T_n}) - V_{T_n-}, \quad n\in\NN.$$
	Hence a jump at a regime switch in $\xi$ implies a drop of the volatility, while a jump in $\eta$ implies a sudden increase. 
\end{remark}

\subsection*{Stationarity of the MSBNS volatility process}

We proceed with establishing conditions for stationarity of the squared volatility process. Recall that the classic BNS volatility process as in \eqref{eq-BNS} is stationary if and only if the driving Lévy process has a finite log-moment. We will see a similar behavior in the Markov switching situation below in Theorem \ref{thm-statMSBNS}. Before, we present a lemma that deals with discrete solutions, i.e. with piecewise constant volatilities. 
Recall that we denote the first return time to $j$ by $\tau_1(j):=\inf\{t > 0: J_t=j, J_{t-}\neq j \}.$

\begin{lemma}\label{lem_MSBNSdegenerate}
There exists a sequence $(c_j)_{j=1,\ldots, N}$ of non-negative constants, such that  under $\PP_\pi$ the squared MSBNS volatility process is strictly stationary with $V_t=c_{J_t}$ $\PP_\pi$-a.s. if and only if
the Lévy processes $(L_t^{(j)})_{t\geq 0}$ are given by 
	\begin{equation} \label{eq-BNSdegenerateLevy} L_t^{(j)}= c_j \cdot t, \quad t\geq 0,\end{equation} 
	and the additional jumps at regime switches fulfill
	\begin{equation}\label{eq-BNSdegeneratejumps}
		e^{-\Delta \xi_{T_n}^{ij}} \big(\Delta \eta_{T_n}^{ij}+c_i \big)=c_j, \quad n\in\NN.
	\end{equation}
\end{lemma}
\begin{proof}
	This follows from \cite[Prop. 4.7]{BS_ExpFuncMAP} upon inserting the definition of the MAP $((U,K),J)$ and noticing that $U$ can only jump at times of regime switches.
\end{proof}

\begin{remark}\rm
	Observe that \eqref{eq-BNSdegeneratejumps} only allows for non-trivial solutions if both components of $(\xi,\eta)$ jump at times of a regime switch. Otherwise, e.g. if $\Delta \xi \equiv 0$, \eqref{eq-BNSdegeneratejumps} simplifies to 
	$$\Delta \eta_{2,\cdot}^{ij} = c_j-c_i, \quad \forall i,j\in S,$$
	and as $\Delta \eta_{2, \cdot}^{ij}\geq 0$ this implies $c_j\geq c_i$ for all pairings $i,j$, and hence $c_j=c_i\equiv c$ and $\Delta \eta_{2, \cdot}^{ij} \equiv 0$. Conversely, if $\Delta \eta_{2, \cdot}^{ij} \equiv 0$ the same argumentation again yields $c_j=c_i=c$ and therefore $\Delta \xi \equiv 0$. 	
\end{remark}

\begin{theorem}\label{thm-statMSBNS}
	Consider the squared MSBNS volatility process $(V_t)_{t\geq 0}$ defined in \eqref{eq-MSBNS} and assume that either \eqref{eq-BNSdegenerateLevy} or \eqref{eq-BNSdegeneratejumps} is not fulfilled. 
	If  
	\begin{equation}
		\int_{(1,\infty)} \log q \,\nu_{L^{(j)}} (\dd q) <\infty \;\text{ and }\; 	\int_{(1,\infty)} \log q \,\PP_j\bigg(\sup_{T_n\leq  \tau_1(j)} e^{-\xi_{2,T_n}} \Delta \eta_{2,T_n} \in \dd q\bigg)<\infty \label{eq-statcondintBNS}
	\end{equation}
	for all $j\in S$, and if one chooses
	$$V_0\overset{d}= V_\infty:= \text{d-}\lim_{t\to\infty} e^{-\xi_t} \int_{(0,t]} e^{\xi_{s-}} \dd K_s,$$
	conditionally independent of $L$ given $J_0$, then $(V_t)_{t\geq 0}$ is strictly stationary. In this case, the corresponding price process $(G_t)_{t\geq 0}$ has stationary increments conditional on $J$, i.e. 
	$$(G_{t+s}-G_t|J_t) \text{ has the same law as } G_s \text{ under } \PP_{J_t} \quad \text{for all $s,t\geq 0$}.$$
\end{theorem}
\begin{proof}
	We follow the same strategy as in the proof of Theorem \ref{thm-statMSCOGARCH} upon noticing that a discrete stationary solution can not occur in the present situation due to Lemma \ref{lem_MSBNSdegenerate}. Thus, in analogy to the proof of Theorem \ref{thm-statMSCOGARCH} it remains to apply \cite[Props. 5.2 and 5.7.1]{BS_ExpFuncMAP}, where in the MSBNS model we note that
	$$	\kappa_\xi= \sum_{j\in S} \pi_j \bigg( \lambda(j)  + \sum_{k\in S\setminus\{j\}}  q_{jk} \int_{\RR_+^2} x \dd F^{jk}(x,y) \bigg)>0$$
	is always fulfilled as $\lambda(j)>0$ by definition. Moreover, in the MSBNS model we observe from the definition \eqref{eq-defAbar} that $\bar{A}_{-\xi^\ast} (x)= \kappa_\xi >0$,
	and in particular $\bar{A}_{-\xi^\ast} (x)$ does not depend on $x$. Thus the integral condition \cite[Eq. (5.10)]{BS_ExpFuncMAP} holds for the MSBNS model if and only if 
	\begin{align} \int_{(1,\infty)} \log q\, \PP_j^\ast  \Big(\sup_{0\leq t \leq \tau_1^\ast(j)} e^{\xi_{t-}^\ast} |\Delta(K_t^{\ast,b} + K_{2,t}^\ast) | \in \dd q\Big)&<\infty, \nonumber \\
	\Leftrightarrow  \int_{(1,\infty)} \log q\, \PP_j \Big(\sup_{0\leq t \leq \tau_1(j)} e^{-\xi_{t-}} |\Delta(K_t^{b} + K_{2,t}) | \in \dd q\Big)&<\infty \label{eq_condstathelp2} \end{align}
	where 
	$$K_t^b=\sum_{\substack{0<s\leq t\\ |\Delta K_s^{(J_s)}|\geq 1 }} \Delta K_s^{(J_s)} =  \sum_{\substack{0<s\leq t\\ |\Delta \eta_s^{(J_s)}|\geq 1 }} \Delta \eta_s^{(J_s)} = \sum_{\substack{0<s\leq t\\ |\Delta L_{s\lambda(J_s)}^{(J_s)}|\geq 1 }} \Delta L_{s\lambda(J_s)}^{(J_s)}.$$
	As $\xi_t>0$ for all $t\geq 0$, and as $\Delta K_{2,T_n} = e^{-\Delta \xi_{T_n}} \Delta \eta_{T_n}$ we hence see that  \eqref{eq_condstathelp2} follows from \eqref{eq-statcondintBNS}.	\\
	Stationarity of the increments of $G$ as stated is now an immediate consequence of \eqref{eq-MSBNS}.
\end{proof}

\begin{remark}\rm \,
	\begin{enumerate}
		\item 	In the special case $|S|=1$ the condition \eqref{eq-statcondintBNS} in Theorem \ref{thm-statMSBNS} reduces to  the well-known necessary and sufficient condition 
		$$\int_{(1,\infty)} \log q \, \nu_L(\dd q)<\infty$$
		for stationarity of a Lévy-driven OU process. 
		\item Note that due to the fact that we incorporated additional jumps at times of a regime switch in our model, the MSBNS volatility process considered in this paper does not fit into the setting of a standard regime-switching Lévy-driven OU process for which stationarity has e.g. been studied in \cite{LindskogMajumder}.
		\item For the classical BNS squared volatility process it has often been highlighted, that the class of possible stationary distributions coincides with the class of selfdecomposable distributions on $\RR_+$. In the Markov switching situation however, we can not expect stationary distributions of $V$ to be selfdecomposable. Hence, a much broader class of stationary distributions can be attained.
	\end{enumerate}

	\end{remark}
	
\subsection*{Moments and autocorrelation structure}

In order to express moment conditions and moments of the squared MSBNS volatility and price process, recall the matrix exponent $\mathbf{\Psi}_\xi$ from \eqref{eq_matrixexp}. In the current situation we see immediately from Definition \ref{defMSBNS} that
$$\mathbf{\Psi}_\xi(-k) = -k \diag(\lambda) + \fatQ^\top \circ \mathbf{F}_{k,0}^\top$$
is finite for all $k\in \NN$. Here, $\mathbf{F}_{k,0}$ is defined according to \eqref{eq-defFmatrix}, which in the MSBNS model specializes to 
\begin{equation*}
	 \mathbf{F}_{k,n}:= \left( \int_{\RR_+^2} e^{-kx} y^n \dd F^{ij}(x,y) \right)_{i,j\in S}. \end{equation*}

We now derive the following proposition along the lines of the proof of Proposition \ref{prop_meanvarianceMSCOGARCH}.
 
 \begin{proposition}\label{prop_meanvarianceMSBNS}
 	Assume that for all $i,j\in S$
 	\begin{equation}\label{eq-condmomentsBNS} \EE_j[|V_0|^k]<\infty,\, \EE[|L_1^{(j)}|^k]<\infty,\, \int_{\RR_+^2} e^{k |x|} \dd F^{ij}(x,y)<\infty,  \text{ and } \int_{\RR_+^2} |y|^{k} \dd F^{ij}(x,y) <\infty,\end{equation} 
 	for $k=1$, then for all $t\geq 0$ and $j\in S$
 	\begin{align*}
 		\EE_j[V_t]&= \bone^\top e^{\mathbf{\Psi}_\xi(-1) \cdot t}\, \mathbf{e}_j \cdot \EE_j[V_0]\\
 		&\qquad  + \bone^\top \int_0^t e^{\mathbf{\Psi}_\xi(-1) \cdot (t-s)} \left( \diag\Big(\lambda(i) \EE[L_1^{(i)}], i\in S\Big) + \fatQ^\top \circ \mathbf{F}_{1,1}^\top \right)e^{\fatQ^\top \cdot s} \mathbf{e}_j \dd s.
 	\end{align*}
 	Furthermore, if \eqref{eq-condmomentsBNS} holds for $k=2$, then for all $0\leq s\leq t$ and $j\in S$
	\begin{align*}
 		\lefteqn{\cov_j(\mathbf{e}_{J_t} V_t, V_s)}\\ &= e^{\mathbf{\Psi}_\xi(-1) (t-s)} \cdot \cov(\mathbf{e}_{J_s} V_s, V_s)\\
 		& +  \left(\sum_{\ell=0}^\infty \frac{(t-s)^\ell}{\ell!} \sum_{k=1}^\ell (\mathbf{\Psi}_\xi(-1))^{k-1} \left( \diag\Big(\lambda(i) \EE[L_1^{(i)}],i\in S\Big) + \fatQ^\top \circ \mathbf{F}_{1,1}^\top \right)(\fatQ^\top)^{\ell-k} \right) \cdot \cov_j(\mathbf{e}_{J_s}, V_s), \end{align*}
 	with
 	$$\cov_j(\mathbf{e}_{J_t}, V_s)  = e^{\fatQ^\top (t-s)} \cov_j(\mathbf{e}_{J_s}, V_s).$$
In particular  
 	$\cov_j(V_t, V_s) = \bone^\top \cov_j(e_{J_t} V_t, V_s) $
 	is decreasing exponentially if the maximal eigenvalue of $\mathbf{\Psi}_\xi(-1)$ is negative.
 \end{proposition} 
 
 We next provide a recursion formula for the integer  moments of the stationary squared volatility. As in Proposition \ref{prop_momentsstat} this formula follows by direct computations from \cite[Thm. 4.8 and Rem. 4.9]{BSdT} upon noticing that $U$ and $K$ have no Gaussian components and only jump simultaneously at time of regime switches.

	\begin{proposition}\label{prop_momentsstatBNS}
	For some $k\in \NN$ and all $j\in S$ assume that
	\begin{enumerate}
		\item[(a)] $\EE[|L_1^{(j)}|^k]<\infty$, 
		\item[(b)] \begin{align}
			\max_{i\in S\setminus\{j\}} \Big((|q_{ii}|+ k \lambda(i) )^{-1} \sum_{\ell \in S\setminus\{i,j\}} q_{i\ell} \int_{\RR_+^2} e^{-kx} \dd F^{i\ell}(x,y)\Big) <1, \label{eq_condpropmomentsstatBNS}
		\end{align}
	\item[(c)] the maximal eigenvalue of $\mathbf{\Psi}_\xi(-k)$ is negative, and
	\item[(d)] all entries of $\mathbf{F}_{k,k}$ are finite.
	\end{enumerate} 
			Then the squared MSBNS volatility process $(V_t)_{t\geq 0}$ has a stationary distribution $V_\infty$, and the $k$'th moment of  $V_\infty$ is given recursively as
		\begin{align*}
			\lefteqn{\EE_\pi[V_\infty^k]}\\
			& = - \fatone^\top \cdot \mathbf{\Psi}_\xi(-k)^{-1}  \cdot  \sum_{n=1}^k {k\choose n} \Big(\diag\Big(\lambda(j) \int_{(1,\infty)} x^n \nu_{L^{(j)}}(\dd x), j\in S \Big) + \fatQ^\top \circ \mathbf{F}_{k,n}^\top \Big) \cdot \fatpi \EE_\pi[V_\infty^{k-n}] ,
		\end{align*}
		with starting value $\EE_\pi[V_\infty^0]=1.$
	\end{proposition}

The elaboration of moments of the increments of the price process $G_t^{(r)}:= G_{t+r}-G_t$ in the MSBNS model is more complicated than in the case of the MSCOGARCH treated in Proposition \ref{prop_momentslogreturns}. Therefore, large parts of the next proposition focus on the martingale part of the price process, i.e. we choose $\mu(j)\equiv\beta(j)\equiv0$, and exclude the leverage term, i.e. we set $\rho(j)\equiv 0$. As in the case of the MSCOGARCH model, and in agreement with the behavior of the BNS model, cf. \cite[Section 4]{BarndorffNielsen-Shephard}, we observe that under these conditions, log-returns are uncorrelated on disjoint time intervals, but squared log-returns are in general correlated.

	\begin{proposition}\label{prop_momentslogreturnsBNS}
		Consider the MSBNS price process $G$ defined in \eqref{eq-MSBNS} and assume that the squared volatility $V$ is strictly stationary. 
		\begin{enumerate}
			\item 	Assume that $\EE[(L_1^{(j)})^2]<\infty$ for all $j$, that all entries of $\mathbf{F}_{0,2}$ are finite, and that the conditions \eqref{eq-condmomentsBNS} are fulfilled for $k=1$. Then for any $t,r\geq 0$ it holds
			\begin{align*}
				\EE_\pi \big[G_t^{(r)}\big]&=\fatone^\top \diag(\mu) \int_0^r e^{\fatQ^\top s} \dd s \cdot \fatpi \\
				& \quad - \fatone^\top \diag(\beta) \mathbf{\Psi}_\xi(-1)^{-1} \Big( \diag\big(\lambda(j) \EE[L_1^{(j)}], j\in S\big) + \fatQ^\top \circ \mathbf{F}_{1,0}^\top \Big) \cdot \fatpi \cdot  r.
			\end{align*}
		\item Assume that $\mu(j)\equiv\beta(j)\equiv\rho(j) \equiv 0$, and  that the conditions \eqref{eq-condmomentsBNS} are fulfilled for $k=1$. Then for any $t\geq 0$ and $h\geq r >0$ it holds
		\begin{align*}
			\EE_\pi\Big[(G_t^{(r)})^2\Big] &= -\fatone^\top  \mathbf{\Psi}_\xi(-1)^{-1} \Big( \diag\big(\lambda(j) \EE[L_1^{(j)}], j\in S\big) + \fatQ^\top \circ \mathbf{F}_{1,0}^\top \Big) \cdot \fatpi \cdot  r, \quad \text{and}\\
			\cov_\pi\big(G_t^{(r)}, G_{t+h}^{(r)}\big) & = 0.
			\end{align*}
			\item Further, under the conditions \eqref{eq-condmomentsBNS} for $k=2$, assuming that $\mu(j)\equiv \beta(j)\equiv \rho(j) \equiv 0$,
		\begin{align*}
			\lefteqn{\cov_\pi\Big((G_t^{(r)})^2, (G_{t+h}^{(r)})^2 \Big)}\\ 
			&= \fatone^\top \mathbf{\Psi}_\xi(-1)^{-1} \Big( e^{h \mathbf{\Psi}_\xi(-1)} - e^{(h-r) \mathbf{\Psi}_\xi(-1)}\Big) \EE_\pi[G_r^2 V_r \mathbf{e}_{J_r}] \\
			& \quad + \fatone^\top \mathbf{\Psi}_\xi(-1)^{-1} \cdot\\
			& \qquad \cdot  \int_0^h \Big(e^{\mathbf{\Psi}_\xi(-1)(h-u)}- e^{\mathbf{\Psi}_\xi(-1)((h-u-r)\vee 0)}\Big) \Big(\diag\big(\lambda(j) \EE[L_1^{(j)}], j\in S\big) + \fatQ^\top \circ \mathbf{F}_{1,1}^\top\Big) e^{\fatQ^\top u} \dd u\\
			&\qquad \cdot 
			\EE_\pi[G_r^2 \mathbf{e}_{J_r}]\\ 
			& \quad - \Big(\fatone^\top  \mathbf{\Psi}_\xi(-1)^{-1} \Big( \diag\big(\lambda(j) \EE[L_1^{(j)}],j\in S\big) + \fatQ^\top \circ \mathbf{F}_{1,0}^\top \Big) \cdot \fatpi \Big)^2 \cdot  r^2.
		\end{align*}
	\end{enumerate}
	\end{proposition}
	\begin{proof}
	By definition of $G$ we have
	\begin{align*}
		\EE_\pi[G_t^{(r)}] &= \EE_\pi\bigg[\int_{(t,t+r]} \mu(J_s) \dd s \bigg]  + \EE_\pi\bigg[\int_{(t,t+r]} \beta(J_s) V_s \dd s \bigg]  + \EE_\pi\bigg[\int_{(t,t+r]} \sqrt{V_s} \dd W_s \bigg]  \\
		& \quad + \EE_\pi\bigg[\int_{(t,t+r]} \rho(J_{s-})\dd \tilde{\eta}_s \bigg] ,
	\end{align*}
where the last two terms are zero due to the integrators being square integrable martingales, and the integrands having finite second moments by our assumptions. Concerning the first term, we note that $(\int_0^\cdot \mu(J_s) \dd s, J)$ is a MAP, and hence we can compute its mean via \cite[Thm. 3.8]{BSdT}  to obtain
\begin{align*} \EE_\pi\bigg[\int_{(t,t+r]} \mu(J_s) \dd s \bigg]& = \EE_\pi\bigg[\EE_{J_t}\bigg[\int_{(0,r]} \mu(J_s) \dd s \bigg]\bigg]
= \EE_\pi\bigg[ \fatone^\top \diag(\mu) \int_0^r e^{\fatQ^\top s} \dd s \cdot \mathbf{e}_{J_t} \bigg] \\
&= \fatone^\top \diag(\mu) \int_0^r e^{\fatQ^\top s} \dd s \cdot \fatpi.\end{align*} 
For the second term, as 
$$\int_{(t,t+r]} \beta(J_s) V_s \dd s = \int_{(t,t+r]}  V_s \dd \Big( \int_0^s \beta(J_u) \dd u\Big),$$
and as $(\int_0^\cdot \beta(J_s) \dd s, J)$ is a MAP, an application of \cite[Lemma 3.11]{BSdT} yields
\begin{align*}
	\EE_\pi\bigg[\int_{(t,t+r]} \beta(J_s) V_s \dd s \bigg] & = \fatone^\top \diag(\beta) \int_t^{t+r} \EE_\pi [V_s \mathbf{e}_{J_s}] \dd s = \fatone^\top \diag(\beta) \EE_\pi[V_0 \mathbf{e}_{J_0}] r,
\end{align*}
where the last equality follows by stationarity of $(V,J)$. Together with Proposition \ref{prop_momentsstatBNS} we may now derive the stated formula for $\EE_\pi[G_t^{(r)}]$.\\
For the second moment we note that
\begin{align*}
		\EE_\pi\Big[(G_t^{(r)})^2\Big]&= \EE_\pi \bigg[\bigg(\int_{(t,t+r]} \sqrt{V_s} \dd W_s \bigg)^2 \bigg]  = \int_{(t,t+r]} \EE_\pi[V_s] \dd s
\end{align*}
from which the stated formula follows via Proposition \ref{prop_momentsstatBNS} due to stationarity of $(V,J)$. \\
The fact that $\cov_\pi(G_t^{(r)}, G_{t+h}^{(r)})=0$ can be shown via It\^o's isometry as we are here only treating the martingale part of the price process.\\
Concerning the squared increments of the martingale part of the price process, a computation similar to the proof of Proposition \ref{prop_momentslogreturns} yields
\begin{align}
	\EE_\pi\Big[ (G_t^{(r)})^2 (G_{t+h}^{(r)})^2\Big] &= \EE_\pi\Big[G_r^2 \EE\big[ (G_{t+h}^{(r)})^2| \tilde{\cF}_r \big]\Big]
	= \EE_\pi\Big[G_r^2 \int_{(h,h+r]} \EE\big[V_s| \tilde{\cF}_r \big] \dd s \Big], \label{eq-BNSlogreturnstep}
\end{align}
where $\tilde{\FF}=(\tilde{\cF}_t)_{t\geq 0}$ is the natural filtration induced by $(V,J,W)$. Now, in analogy to the proof of Proposition \ref{prop_momentslogreturns} and \cite[Proof of Prop. 5.1]{KLM2004}, 
\begin{align*}
	\EE[V_s|\tilde{\cF}_r] &= \EE_{J_r}[V_{s-r}] + \big(V_r-\EE[V_0]\big) \EE_{J_r}\big[e^{-\xi_{s-r}}|\tilde{\cF}_0\big]\\
	&= \fatone^\top e^{\mathbf{\Psi}_\xi(-1) (s-r)} \mathbf{e}_{J_r} V_r \\ & \quad + \fatone^\top \int_0^{s-r} e^{\mathbf{\Psi}_\xi(-1)(s-r-u)} \Big(\diag\big(\lambda(j)\EE[L_1^{(j)}],j\in S\big) + \Big(\fatQ^\top \circ \mathbf{F}_{1,1}^\top \Big)\Big) e^{\fatQ^\top u} \mathbf{e}_{J_r} \dd u,
\end{align*}
and inserting this in \eqref{eq-BNSlogreturnstep} yields the result.
	\end{proof}

\subsection*{A special case with no jumps in the $\xi$-component}

As jumps at a regime switch in the the $\xi$-component lead to a sudden drop in the volatility process or at least dampen the increase induced by a simultaneous jump in $\eta$, and as typical volatility jumps are positive, it seems reasonable to consider the special case of an MSBNS model under the assumption of no additional jumps in $\xi$. Following Remark \ref{rem-defMSBNS} in this case we observe only upward jumps of the volatility at times of a regime switch. In particular we then have 
$$\Delta V_{T_n}=\Delta \eta_{T_n}, \quad n\in\NN$$ 
and hence these exogeneous jumps can be modeled directly via the distributions $F^{ij}$. \\
Moreover, this assumption leads to various other simplifications of the model:
\begin{itemize}
	\item Concerning stationarity we observe from Lemma \ref{lem_MSBNSdegenerate} that a discrete solution is only possible if and only if $\Delta \eta^{ij} _{T_n} = c_j-c_i$ for all $n$, i.e. for deterministic jump heights.\\
	Excluding the discrete solutions from Theorem \ref{thm-statMSBNS} we observe that a stationary solution exists if for all $j\in S$
	\begin{equation*}
		\int_{(1,\infty)} \log q \,\nu_{L^{(j)}} (\dd q) <\infty\;  \text{  and } \;	\int_{(1,\infty)} \log q \,\PP_j\bigg(\sup_{T_n\leq  \tau_1(j)}  \Delta \eta_{2,T_n} \in \dd q\bigg)<\infty, 
	\end{equation*}
that is, if all driving Lévy processes and the distributions $F^{ij}$ have a finite log-moment.
\item In the computations of the moments we note that 	$\mathbf{\Psi}_\xi$ is diagonal, as $\mathbf{F}_{k,0}= \mathbf{0}$. In particular we observe that $e^{\mathbf{\Psi}_\xi(-1)} = \diag(e^{-\lambda(j)},j\in S)$ and to illustrate the resulting simplifications, we note by example that from Proposition \ref{prop_meanvarianceMSBNS}
	\begin{align*}
	\EE_j[V_t]&= e^{-\lambda(j)} \EE[V_0] \\
	& \quad + \bone^\top \int_0^t \diag\big(e^{-\lambda(i) (t-s)}, i\in S\big)  \left( \diag\big(\lambda(i) \EE[L_1^{(i)}], i\in S\big) + \fatQ^\top \circ \mathbf{F}_{1,1}^\top \right)e^{\fatQ^\top \cdot s} \mathbf{e}_j \dd s,
\end{align*}
while from Proposition \ref{prop_momentslogreturnsBNS}
	\begin{align*}
	\EE_\pi \big[G_t^{(r)}\big]&=\fatone^\top \diag(\mu) \int_0^r e^{\fatQ^\top s} \dd s \cdot \fatpi  + \fatone^\top \diag\big(\beta(j) \EE[L_1^{(j)}], j\in S\big)  \cdot \fatpi \cdot  r.
\end{align*}
Moreover, in this case $\mathbf{\Psi}_\xi(-k)$ has a negative maximal eigenvalue for all $k\in\NN$ and in particular $\cov_j(V_t, V_s) $ is always decreasing exponentially as $t-s$ grows.
\end{itemize}

\section{Discussion and outlook}\label{S-end}

We have proposed two Markov switching volatility models generalizing the well-known BNS and COGARCH model to a random environment. Despite the newly gained flexibility the resulting models remain mathematically tractable and can be analyzed nicely using the recently established theory on MMGOU processes.\\

While the classical COGARCH model only depends on a single source of randomness (the driving Lévy process $L$), the MSCOGARCH model additionally depends on the environment, modeled by the continuous-time Markov chain $J$, and the exogeneous shocks at times of a regime switch. Still, basic features of the COGARCH, such as the existence of a stationary version of the volatility process, a recursion formula for the moments of the volatility process, and the autocorrelation structure of the increments of the price process are kept.\\

This is similar in case of the MSBNS model which also inherits many properties of the BNS model. While the BNS model in its original form already depends on two independent sources of randomness (the driving subordinator $L$ and the Brownian motion), the MSBNS model  incorporates dependence of the driving process and the dynamics (modeled by $\lambda$) on the environment. Further, exogeneous shocks at times of a regime switch allow for an even higher flexibility in the model, while basic stylized features of financial data are kept.\\

Another well-known property of the BNS and COGARCH volatility processes is heaviness of the tails of the stationary distribution under certain conditions. In particular, the COGARCH volatility and price process admit Pareto-like tails under weak moment conditions, see \cite[Thm.~6]{KLM06}. This can also be expected in the case of the MSCOGARCH volatility, as by construction it fulfills a random recurrence equation of the form 
$V_t= A_{s,t} V_s + B_{s,t}.$ More precisely, considering the return times $\tau_n(j)$ of the Markov chain $J$, under $\PP_j$ it holds 
$$V_{\tau_n(j)} = e^{-(\xi_{\tau_n(j)} - \xi_{\tau_{n-1}(j)})} V_{\tau_{n-1}(j)} + e^{-\xi_{\tau_n(j)} }\int_{(\tau_{n-1}(j), \tau_n(j)]} e^{\xi_{s-}} \dd \eta_s =: A_n V_{\tau_{n-1}(j)} + B_n,$$
for an i.i.d. sequence $(A_n,B_n)_{n\in\NN}$. Hence applying the classical results on tails of perpetuities by Kesten \cite{Kesten} and Goldie \cite{Goldie_TailsiidRRE} yields conditions for heavy tailed volatility distributions in the MSCOGARCH model. This rough sampling only at return times however will not allow for a deeper insight into the contributions of different regimes to the tail behaviour. Hence, a deeper study has to rely on tails of perpetuities in a random environment which so far have only be considered in some special cases (e.g. in \cite{roitershtein}), that do not fit the MSCOGARCH or the MSBNS model. We therefore refrain from any details here.

\section*{Acknowledgements}

My thanks go to M.Sc. Markus Vogelsang for providing me with a Python program library for the simulation of various MMGOU processes.

\end{document}